\documentclass[12pt]{article}

\usepackage[english]{babel}
\usepackage[T1]{fontenc}
\usepackage[utf8]{inputenc}

\def\squareforqed{\hbox{\rlap{$\sqcap$}$\sqcup$}}
\def\qed{\ifmmode\squareforqed\else{\unskip\nobreak\hfil
\penalty50\hskip1em\null\nobreak\hfil\squareforqed
\parfillskip=0pt\finalhyphendemerits=0\endgraf}\fi}

\usepackage{amsmath}
\usepackage{amssymb}

\usepackage{xspace}
\usepackage{amsfonts}
\usepackage{amssymb}
\usepackage{mathtools}
\usepackage{amsthm}
\usepackage{stmaryrd}
\usepackage{tabularx}
\usepackage{tikz}
\usetikzlibrary{positioning}
\usetikzlibrary{shapes}
\usetikzlibrary{arrows}
\usetikzlibrary{arrows.meta}
\usepackage{todonotes}
\usepackage[draft]{hyperref}
\usepackage{caption}
\usepackage{pdflscape,array,booktabs}
\usepackage{multicol}
\usepackage[utf8]{inputenc}
\usepackage[english]{babel}
\usepackage[T1]{fontenc}
\usepackage{csquotes}
\usepackage{wrapfig}

\newcommand{\alc}{\ensuremath{\mathcal{ALC}}\xspace}
\newcommand{\el}{\ensuremath{\mathcal{E\!L}}\xspace}

\newcommand{\I}{\ensuremath{\mathcal{I}}\xspace}

\newcommand{\On}{\ensuremath{\mathfrak{O}}\xspace}
\newcommand{\Ons}{\ensuremath{\mathfrak{O}_s}\xspace}
\newcommand{\Onr}{\ensuremath{\mathfrak{O}_r}\xspace}
\newcommand{\T}{\ensuremath{\mathcal{T}}\xspace}
\newcommand{\A}{\ensuremath{\mathcal{A}}\xspace}

\newcommand{\Con}{\textit{Con}}
\newcommand{\NI}{\ensuremath{N_I}\xspace}
\newcommand{\NC}{\ensuremath{N_C}\xspace}
\newcommand{\NR}{\ensuremath{N_R}\xspace}
\newcommand{\susu}{\succ^{\mathit{sub}}}
\newcommand{\Upper}{\textit{Upper}}
\newcommand{\susy}{\succ^{\mathit{syn}}}
\newcommand{\subeq}{\sqsubseteq^\emptyset}
\newcommand{\sub}{\sqsubset^\emptyset}

\newcommand{\sqsubsyn}{\sqsubset^{\mathit{syn}}}
\newcommand{\sqsubsyneq}{\sqsubseteq^{\mathit{syn}}}
\newcommand{\sqsupsyn}{\sqsupset^{\mathit{syn}}}

\newcommand{\Tmc}{\ensuremath{\mathcal{T}}\xspace}
\newcommand{\Vmc}{\ensuremath{\mathcal{V}}\xspace}
\newcommand{\Wmc}{\ensuremath{\mathcal{W}}\xspace}

\newtheorem%
     {theorem}{Theorem}

\newtheorem%
     {corollary}[theorem]{Corollary}

\newtheorem%
     {proposition}[theorem]{Proposition}

\newtheorem%
     {lemma}[theorem]{Lemma}

\newtheorem%
     {definitionAux}[theorem]{Definition}
\newenvironment%
     {definition}{\begin{definitionAux}\rm}{\end{definitionAux}}

\newtheorem%
     {remarkAux}[theorem]{Remark}
\newenvironment%
     {remark}{\begin{remarkAux}\rm}{\end{remarkAux}}

\newtheorem%
     {exampleAux}[theorem]{Example}

\newenvironment%
{example}{\begin{exampleAux}\rm}{\end{exampleAux}}

\newtheorem%
     {examplesAux}[theorem]{Examples}

\newenvironment%
     {examples}{\begin{examplesAux}\rm}{\end{examplesAux}}

\title{Repairing Description Logic Ontologies by Weakening Axioms}

        \author{Franz Baader 
                Francesco Kriegel 
                Adrian Nuradiansyah,
                and Rafael Pe{\~n}aloza 
        }

\title{Repairing Description Logic Ontologies by Weakening Axioms}

        \author{Franz Baader,
                 Francesco Kriegel,
                 Adrian Nuradiansyah,\\
                 Rafael Pe{\~n}aloza
        }

\begin{document}

%
%
%
\date{}
\maketitle

\begin{abstract}
                The classical approach for repairing a Description Logic ontology $\On$ in the sense of removing
                an unwanted consequence $\alpha$ is 
                to delete a minimal number of axioms from $\On$ such that the resulting ontology $\On'$ does not have the consequence $\alpha$.
                However, the complete deletion of axioms may be too rough, in the sense that it may also remove consequences that are
                actually wanted.
                To alleviate this problem, we propose a more gentle way of repair in which axioms are not necessarily
                deleted, but only weakened. On the one hand, we investigate general properties of this
                gentle repair method. On the other hand, we propose and analyze concrete approaches for
                weakening axioms expressed in the Description Logic  $\el$.
\end{abstract}

\section{Introduction}
\label{sec:introduction}

Description logics (DLs) \cite{BCNMP03,DBLP:books/daglib/0041477} are a family of logic-based
knowledge representation formalisms, which are employed in various
application domains,
such as natural language processing, configuration,
databases, and bio-medical ontologies,
but their most notable
success so far is the adoption of the DL-based language
OWL\footnote{%
see \href{https://www.w3.org/TR/owl2-overview/}{https://www.w3.org/TR/owl2-overview/} for its most recent edition OWL\,2.}
as standard ontology language for the Semantic Web.
As the size of DL-based ontologies grows, tools that support
improving the quality of such ontologies become more important.
DL reasoners\footnote{%
see \href{http://owl.cs.manchester.ac.uk/tools/list-of-reasoners/}{http://owl.cs.manchester.ac.uk/tools/list-of-reasoners/}}
can be used to detect inconsistencies and to
infer other implicit consequences,
such as subsumption and instance relationships.
However, for the developer of a DL-based ontology, it is often quite hard to
understand why a consequence computed by the reasoner
actually follows from the knowledge base, and how to repair the ontology in case
this consequence is not intended.
 
Axiom pinpointing \cite{ScCo03} was introduced to help developers or users
of DL-based ontologies understand the reasons why a certain consequence holds by
computing so-called \emph{justifications}, i.e., minimal subsets of the ontology that have the consequence in question.
Black-box approaches for computing justifications such as \cite{SHCH07,KPHS07,BaSu08}
use repeated calls of existing highly-optimized DL reasoners for this purpose, but it may be necessary
to call the reasoner an exponential number of times. In contrast, glass-box approaches such as \cite{BaHo95,ScCo03,PaSK05,MLBP06}
compute all justifications by a single run of a modified, but usually less efficient reasoner.

Given all justifications of an unwanted consequence, one can then repair the ontology by removing one axiom
from each justification. However, removing complete axioms may also eliminate consequences that are actually wanted.
For example, assume that our ontology contains the following terminological axioms:
$$
\begin{array}{l}
\mathit{Prof} \sqsubseteq \exists\mathit{employed}.\mathit{Uni}\sqcap 
                          \exists\mathit{enrolled}.\mathit{Uni},\\[.3em]
\exists\mathit{enrolled}.\mathit{Uni} \sqsubseteq \mathit{Studi}.
\end{array}
$$
These two axioms are a justification for the incorrect consequence that professors are students. While the
first axiom is the culprit, removing it completely would also remove the correct consequence that
professors are employed by a university. Thus, it would be more appropriate to replace the first axiom 
by the weaker axiom $\mathit{Prof} \sqsubseteq \exists\mathit{employed}.\mathit{Uni}$. This is the basic idea
underlying our \emph{gentle repair approach}. In general, in this approach we weaken one axiom from each justification
such that the modified justifications no longer have the consequence.

Approaches for repairing ontologies while keeping more consequences than the classical approach based on
completely removing axioms have already been considered in the literature. On the one hand, there are approaches that
first modify the given ontology, and then repair this modified ontology using the
classical approach. In \cite{Horridge11}, a specific syntactic structural transformation
is applied to the axioms in an ontology, which replaces them by sets of logically weaker axioms.
More recently, the authors of \cite{DuQF14} have generalized this idea by allowing for different specifications
of the structural transformation of axioms. They also introduce a specific structural transformation that is based on specializing 
left-hand sides and generalizing right-hand sides of axioms in a way that ensures finiteness of the obtained set of axioms.
Closer to our gentle repair approach is the one in
\cite{Lam08}, which adapts the tracing technique from 
\cite{Baader1995} to identify not only the axioms that cause a consequence,
but also the parts of these axioms that are actively involved in deriving the consequence. This provides
them with information for how to weaken these axioms. In \cite{TC+18}, repairs are computed by weakening axioms with the
help of refinement operators that were originally introduced for the purpose of concept learning \cite{LeHi10}.

In this paper, we will introduce a general framework for repairing ontologies based on axiom weakening.
This framework is independent of the concrete method employed for weakening axioms and of the concrete ontology language used
to write axiom. It only assumes that ontologies are finite sets of axioms, that there is a monotonic consequence operator 
defining which axiom follows from which, and that weaker axioms have less consequences.
However, all our examples will consider ontologies expressed in the light-weight DL $\el$. Our first important
result is that, in general, the gentle repair approach needs to be iterated, i.e., applying it once does not
necessarily remove the consequence. This problem has actually been overlooked in \cite{Lam08}, which means that their
approach does not always yield a repair. Our second result is that at most exponentially many iterations are always
sufficient to reach a repair. The authors of \cite{TC+18} had already realized that iteration is needed, but
they did not give an example explicitly demonstrating this, and they had no termination proof. Instead of allowing
for arbitrary ways of weakening axioms, we then introduce the notion of a \emph{weakening relation}, which restricts
the way in which axioms can be weakened. Subsequently, we define conditions on such weakening relations that equip
the gentle repair approach with better algorithmic properties if they are satisfied. Finally, we address the task of defining specific
weakening relations for the DL $\el$. After showing that two quite large such relations do not behave well,
we introduce two restricted relations, which are based on generalizing the right-hand sides of axioms semantically 
or syntactically. Both of them satisfy most of our conditions, but from a complexity point of view the syntactic
variant behaves considerably better.



\section{Basic definitions}

In the first part of this section, we introduce basic notions from DLs to provide us with concrete
examples for how ontologies and their axioms may look like. In the second part, we provide basic definitions
regarding the repair of ontologies, which are independent of the ontology language these ontologies are written in.
However, the concrete examples given there are drawn from DL-based ontologies.

\subsection{Description Logics}
\label{subsec:dls}

A wide range of DLs of different expressive power haven  been investigated in the literature.
Here, we only introduce the DL \el, for which reasoning is tractable~\cite{Bran04}.

Let \NC and \NR be mutually disjoint 
sets of \emph{concept} and \emph{role names}, respectively. 
Then \emph{\el concepts} over these names are constructed through the grammar rule 
$$
C::= A \mid \top\mid C \sqcap C\mid \exists r.C,
$$
where $A\in\NC$ and $r\in\NR$, i.e., the DL \el has the concept constructors $\top$ (top concept), $\sqcap$ (conjunction),
and $\exists r.C$ (existential restriction). 
The \emph{size} of an $\el$ concept $C$ is the number of occurrences of $\top$ as well as concept and role names in $C$, 
and its \emph{role depth} is the maximal nesting of existential restrictions. If $S$ is a finite set of $\el$ concepts, then we denote the conjunction
of these concepts as $\bigsqcap S$.

Knowledge is represented using appropriate axioms formulated using
concepts, role names and an additional set of individual names \NI. An \emph{\el axiom} is either a \emph{GCI} of the form $C\sqsubseteq D$ with $C,D$ concepts, or an
\emph{assertion}, which is of the form $C(a)$ (\emph{concept assertion}) or $r(a,b)$ (\emph{role assertion}), 
with $a,b\in\NI, r\in\NR$, and $C$ a concept.
A finite set of GCIs is called a \emph{TBox}; a finite set of assertions is an \emph{ABox}.  An \emph{ontology} is
a finite set of axioms.

The semantics of \el is defined using \emph{interpretations}
$\I=(\Delta^\I,\cdot^\I)$, where $\Delta^\I$ is a non-empty set, called the \emph{domain},
and $\cdot^\I$ is the \emph{interpretation function}, which maps every $a\in\NI$ to an element $a^\I\in\Delta^\I$,
every $A\in\NC$ to a set $A^\I\subseteq\Delta^\I$, and every $r\in\NR$ to a binary relation 
$r^\I\subseteq\Delta^\I\times\Delta^\I$. The interpretation function $\cdot^\I$ is extended to arbitrary \el concepts by setting
$\top^\I := \Delta^\I$,
$(C\sqcap D)^\I:=C^\I\cap D^\I$, and
$(\exists r.C)^\I:=\{\delta\in\Delta^\I\mid \exists \eta\in C^\I. (\delta,\eta)\in r^\I\}$. 
  
The interpretation \I \emph{satisfies} the GCI $C\sqsubseteq D$ if $C^\I\subseteq D^\I$; it \emph{satisfies}
the assertion $C(a)$ and $r(a,b)$, if $a^\I\in C^\I$ and $(a^\I,b^\I)\in r^\I$, respectively.
It is a \emph{model} of the TBox \T, the ABox \A, and the ontology \On, if it satisfies all the axioms in \T, \A, and
\On, respectively.
Given an ontology \On, and an axiom $\alpha$, we say that $\alpha$ is a \emph{consequence} of \On (or that
\On \emph{entails} $\alpha$) if every model of \On satisfies $\alpha$. In this case, we write
$\On\models\alpha$. The \emph{set of all consequences} of \On is denoted by \Con(\On). 
As shown in \cite{Bran04}, consequences in $\el$ can be decided in polynomial time.
We say that the two axioms $\gamma,\delta$ are \emph{equivalent} if
$\Con(\{\gamma\})=\Con(\{\delta\})$.

A \emph{tautology} is an axiom $\alpha$ such that $\emptyset\models\alpha$, where
$\emptyset$ is the ontology that contains no axioms. For example, GCIs of the form $C\sqsubseteq \top$ and $C\sqsubseteq C$, and
assertions of the form $\top(a)$ are tautologies. We write $C\subeq D$ to indicate that the GCI $C\sqsubseteq D$ is a tautology.
In this case we say that $C$ is \emph{subsumed} by $D$.  We say that the concepts $C, D$ are \emph{equivalent} (written $C\equiv^\emptyset D$) 
if $C\subeq D$ and $D\subeq C$; and that 
$C$ is \emph{strictly subsumed} by $D$ (written $C\sub D$) if $C\subeq D$ and $C\not\equiv^\emptyset D$.

The following recursive characterization of the subsumption relation $\subeq$ has been proved in \cite{BaMo10}.

\begin{lemma}\label{subs:char:lem}
Let $C,D$ be two \el concepts such that
\begin{align*}
C = {} & A_1\sqcap\ldots\sqcap A_k\sqcap\exists r_1.C_1\sqcap\ldots\sqcap \exists r_m.C_m \\
D = {} & B_1\sqcap\ldots\sqcap B_\ell\sqcap\exists s_1.D_1\sqcap\ldots\sqcap \exists s_n.D_n,
\end{align*}
and $A_1,\ldots, A_k, B_1,\ldots,B_\ell\in\NC$.
Then $C\subeq D$ iff $\{B_1,\ldots, B_\ell\}\subseteq\{A_1,\ldots,A_k\}$ and
for every $j, 1\leq j\leq n$, there exists an $i,1\leq i\leq m$, such that $r_i = s_j$
and $C_i\subeq D_j$.
\end{lemma}



\subsection{Repairing Ontologies}
\label{subsec:repairing}

For the purpose of this subsection and also large parts of the rest of this paper, we leave it open what sort 
of axioms and ontologies are allowed in general, but we draw our examples from $\el$ ontologies. 
We only assume that there is a monotonic consequence relation $\On\models\alpha$
between ontologies (i.e., finite sets of axioms) and axioms, and that $\Con(\On)$ consists of all 
consequences of \On.


Assume in the following that the ontology $\On = \Ons\cup\Onr$ is the disjoint union of a \emph{static} ontology \Ons and a \emph{refutable}
ontology \Onr. When repairing the ontology, only the refutable part may be changed. For example, the static part of the ontology could be a carefully
hand-crafted TBox whereas the refutable part is an ABox that is automatically generated from (possibly erroneous) data. It may also make sense to
classify parts of a TBox as refutable, for example if the TBox is obtained as a combination of ontologies from different sources, some of which may be
less trustworthy than others. In a privacy application \cite{Grau2010PrivacyIO,BaBN17}, 
it may be the case that parts of the ontology are publicly known whereas other parts
are hidden. In this setting, in order to hide critical information, it only makes sense to change the hidden part of the ontology.

\begin{definition}
	Let $\On  = \Ons\cup\Onr$ be an ontology consisting of a static and a refutable part, and 
	$\alpha$ an axiom such that $\On\models\alpha$ and $\Ons\not\models\alpha$.
	The ontology $\On'$ is a \emph{repair of $\On$ w.r.t.\ $\alpha$} if
	$$
	\Con(\Ons\cup\On') \subseteq \Con(\On)\setminus\{\alpha\}.
	$$
	The repair $\On'$ is an \emph{optimal repair} of $\On$ w.r.t.\ $\alpha$ if there is no 
	repair $\On''$ of $\On$ w.r.t.\ $\alpha$ with $\Con(\Ons\cup\On')\subset\Con(\Ons\cup\On'')$.
	The repair $\On'$ is a \emph{classical repair} of $\On$ w.r.t.\ $\alpha$ if $\On'\subset\Onr$,
	and it is an \emph{optimal classical repair} of $\On$ w.r.t.\ $\alpha$ if
	there is no classical repair $\On''$ of $\On$ w.r.t.\ $\alpha$ such that $\On'\subset\On''$.
\end{definition}

The condition $\Ons\not\models\alpha$ ensures that \On does have a repair w.r.t.\ $\alpha$ since
obviously the empty ontology $\emptyset$ is such a repair. In general, \emph{optimal} repairs need not exist.

\begin{proposition}
There is an \el ontology $\On  = \Ons\cup\Onr$ and an \el axiom $\alpha$ such that
$\On$ does not have an optimal repair w.r.t.\ $\alpha$.
\end{proposition}

\begin{proof}
        We set $\alpha := A(a)$, $\Ons := \T$, and $\Onr :=\A$ where
	$$
	\begin{array}{l}
	\T := \{ A \sqsubseteq \exists r.A,   \exists r.A \sqsubseteq A\}\ \ \mbox{and}\ \ 
	\A :=  \{A(a)\}.
	\end{array}
	$$
	%
	To show that there is no optimal repair of $\On$ w.r.t.\ $\alpha$,
	we consider an arbitrary repair $\On'$ and show that it cannot be optimal.
        Thus, let $\On'$ be such that $$\Con(\T\cup \On') \subseteq \Con(\On)\setminus\{A(a)\}.$$ 
	Without loss of generality we assume that $\On'$ contains assertions only. In fact, if $\On'$ contains a
	GCI that does not follow from $\T$, then $\Con(\T\cup \On') \not\subseteq \Con(\On)$. This is an easy consequence of
	the fact that, in $\el$, a GCI follows from a TBox together with an ABox iff it follows from the TBox alone.
	It is also easy to see that $\On'$ cannot contain role assertions 
        since no such assertions are entailed by \On.
	In addition, concept assertions following from $\T\cup \On'$ must have a specific form.\\[.5em]
	\noindent
	\emph{\textbf{Claim:}} If the assertion $C(a)$ is in $\Con(\T\cup \On')$, then $C$ does not contain $A$.\\[.5em]
	\noindent
	\emph{Proof of claim.} By induction on the role depth $n$ of $C$.\\[.3em]
	\noindent
	\emph{Base case:}
	If $n = 0$ and $A$ is contained in $C$, then $A$ is a conjunct of $C$ and thus
	$C(a)\in \Con(\T\cup \On')$ implies $A(a)\in \Con(\T\cup \On')$, which is a contradiction.\\[.3em]
	\noindent
	\emph{Step case:}
	If $n > 0$ and $A$ occurs at role depth $n$ in $C$, then $C(a)\in \Con(\T\cup \On')$ implies 
	that there are roles $r_1,\ldots, r_n$ such that $(\exists r_1.\cdots\exists r_n.A)(a)\in \Con(\T\cup \On').$
	Since $\Con(\T\cup \On') \subseteq \Con(\On)$, this can only be the case if $r_1 = \ldots = r_n = r$ since
	$\On$ clearly has models in which all roles different from $r$ are empty. Since $\T$ contains the 
	GCI $\exists r.A \sqsubseteq A$ and $r_n = r$, $(\exists r_1.\cdots\exists r_n.A)(a)\in \Con(\T\cup \On')$
	implies $(\exists r_1.\cdots\exists r_{n-1}.A)(a)\in \Con(\T\cup \On').$ Induction now yields that this is
	not possible, which completes the proof of the claim. 
	
	\smallskip
	Furthermore, as argued in the proof of the claim, any assertion belonging to $\Con(\On)$ cannot
	contain roles other than $r$. The same is true for concept names different from $A$. Consequently, all assertions
	$C(a)\in \Con(\T\cup \On')$ are such that $C$ is built using $r$ and $\top$ only. Any such concept $C$ is
	equivalent to a concept of the form $(\exists r.)^n\top$.
	
	Since $\On'$ is finite, there is a maximal $n_0$ such that $((\exists r.)^{n_0}\top)(a)\in \On'$, but
	$((\exists r.)^{n}\top)(a)\not\in \On'$ for all $n > n_0$. Since $(\exists r.)^n\top\sqsubseteq (\exists r.)^m\top$ if $m\leq n$,
	we can assume without loss of generality that $\On' = \{ ((\exists r.)^{n_0}\top)(a) \}$. 
%
%
	We claim that $((\exists r.)^{n}\top)(a)\not\in \Con(\T\cup \On')$ if $n > n_0$. To this purpose, we construct a model
	$\I$ of $\T\cup \On'$ such that $a^\I \not\in ((\exists r.)^{n}\top)^\I$. This model is defined as follows:
	$$
	\begin{array}{lll}
	\Delta^\I & = & \{d_0,d_1,\ldots,d_{n_0}\}, \\
	r^\I & = & \{(d_{i-1},d_i) \mid 1\leq i\leq n_0\},\\
	A^\I & = & \emptyset,\\
	a^\I & = & d_0.
	\end{array}
	$$
%
	Clearly, $\I$ is a model of $\On'$, and it does not satisfy $((\exists r.)^{n}\top)(a)$ if $n > n_0$. In addition,
	it is a model of $\T$ since $A^\I = (\exists r.A)^\I = \emptyset$.
	
	Consequently, if we choose $n$ such that $n > n_0$ and define $\On'' := \{((\exists r.)^{n}\top)(a)\}$, then $\Con(\T\cup \On') \subset Con(\T\cup \On'')$. 
	In addition, $\Con(\T\cup \On'') \subseteq \Con(\On)\setminus\{A(a)\}$, i.e., $\On''$ is a repair.
	%
        This shows that $\On'$ is not optimal. Since we have chosen
	$\On'$ to be an arbitrary repair, this shows that there cannot be an optimal repair.
\end{proof}

In contrast, optimal \emph{classical} repairs always exist.
One approach for computing such a repair uses justifications and hitting sets~\cite{Reit87}.

\begin{definition}
	Let $\On = \Ons\cup\Onr$ be an ontology and 
	$\alpha$ an axiom such that $\On\models\alpha$ and $\Ons\not\models\alpha$.
	A \emph{justification} for $\alpha$ in $\On$ is a minimal subset $J$ of $\Onr$ such that $\Ons\cup J\models \alpha$.
	Given justifications $J_1,\ldots, J_k$ for $\alpha$ in $\On$, a \emph{hitting set} of these justifications
	is a set $H$ of axioms such that $H\cap J_i\neq \emptyset$ for $i = 1,\ldots,k$. This hitting set is \emph{minimal}
	if there is no other hitting set strictly contained in it.
\end{definition}
Note that the condition $\Ons\not\models\alpha$ implies that justifications are non-empty. Consequently,
hitting sets and thus minimal hitting sets always exist.

The \emph{algorithm for computing an optimal classical repair} of $\On$ w.r.t.\ $\alpha$ proceeds in two steps:
(i)~compute all justifications $J_1,\ldots, J_k$ for $\alpha$ in $\On$;
and then
(ii)~compute a minimal hitting set $H$ of $J_1,\ldots, J_k$ and remove the elements of $H$ from $\Onr$, i.e.,
output $\On' = \Onr \setminus H$.

It is not hard to see that, independently of the choice of the hitting set, this algorithm produces an
optimal classical repair. Conversely, all optimal classical repairs can be generated this way by going through all hitting sets. 
%
%


\section{Gentle Repairs}
\label{sec:gentle}
Instead of removing axioms completely, as in the case of a classical repair,
a gentle repair replaces them by weaker axioms.

\begin{definition}
	\label{weaker:def}
	Let $\beta, \gamma$ be two axioms. We say that $\gamma$ is \emph{weaker than $\beta$} if
	$\Con(\{\gamma\})\subset\Con(\{\beta\})$.
\end{definition}
Alternatively, we could have introduced \emph{weaker w.r.t\ the strict part of the ontology}, by requiring 
$\Con(\Ons\cup\{\gamma\})\subset\Con(\Ons\cup\{\beta\})$.\footnote{%
	Defining weaker w.r.t\ the whole ontology $\On$ does not make sense since this ontology is possibly erroneous.
}
In this paper, we will not consider this alternative definition, although most of the results in this section
would also hold w.r.t.\ it (e.g., Theorem~\ref{always:term:prop}). The difference between the two definitions
is, however, relevant in the next section, where we consider concrete approaches for how to weaken axioms.
In  the case where the whole ontology is refutable, there is of course no difference between the two definitions.

Obviously, the \emph{weaker-than} relation from Definition~\ref{weaker:def} is transitive, i.e., if 
$\alpha$ is weaker than $\beta$ and $\beta$ is weaker than $\gamma$, then $\alpha$ is also weaker than $\gamma$.
In addition, a tautology is always weaker than a non-tautology. Replacing an axiom by a tautology is obviously the same
as removing this axiom. We assume in the following that there exist tautological axioms, which is obviously true for
description logics such as \el. 
\paragraph{Gentle repair algorithm:}
we still compute all justifications $J_1,\ldots, J_k$ for $\alpha$ in $\On$ and a minimal hitting set $H$ of $J_1,\ldots, J_k$.
But instead of removing the elements of $H$ from $\Onr$, we replace them by weaker axioms. To be more precise,
if $\beta\in H$ and $J_{i_1},\ldots,J_{i_\ell}$ are all the justifications containing $\beta$, then replace $\beta$ by
a weaker axiom $\gamma$ such that 
\begin{equation}\label{replace:prop}
\Ons\cup(J_{i_j}\setminus \{\beta\})\cup\{\gamma\} \not\models \alpha\ \ \mbox{for $j = 1,\ldots,\ell$}.
\end{equation}
Note that such a weaker axiom $\gamma$ always exists. In fact, we can choose a tautology as the axiom $\gamma$. If
$\gamma$ is a tautology, then replacing $\beta$ by $\gamma$ is the same as removing $\beta$. Thus, we have
$\Ons\cup(J_{i_j}\setminus \{\beta\})\cup\{\gamma\} \not\models \alpha$ due to the minimality of $J_{i_j}$. In addition,
minimality of $J_{i_j}$ also implies that $\beta$ is not a tautology since otherwise $\Ons \cup (J_{i_j}\setminus \{\beta\})$
would also have the consequence $\alpha$. In general, different choices of $\gamma$ yield different runs of the 
algorithm.

\vspace*{.5em}
In principle, the algorithm could always use a tautology $\gamma$, but then this run would produce a
classical repair. To obtain more gentle repairs, the algorithm needs to use a strategy that chooses stronger axioms (i.e.,
axioms $\gamma$ that are less weak than tautologies) if possible. 
%
In contrast to what is claimed in the literature (e.g.~\cite{Lam08}), this approach does not necessarily yield a repair.

\begin{lemma}\label{iteration:needed:lem}
Let $\On'$ be the ontology obtained from $\Onr$ by replacing all the elements of the hitting set
by weaker ones such that the condition (\ref{replace:prop}) is satisfied. 
Then $\Con(\Ons\cup\On')\subseteq \Con(\On)$, but in general
we may still have $\alpha\in \Con(\Ons\cup\On')$.
\end{lemma}

\begin{proof}
The definition of ``weaker than'' (see Definition~\ref{weaker:def}) obviously implies that 
$\Con(\Ons\cup\On')\subseteq \Con(\On)$. 

We now give an example where this approach nevertheless does not produce a repair.
	Let $\On = \Ons \cup \Onr$ where $\Ons = \emptyset$ and $\Onr = \T \cup \A$ with
	$\T = \{ B\sqsubseteq A \}$ and
	$\A = \{ (A\sqcap B)(a) \}$, 
	and $\alpha$ be the consequence $A(a)$. Then $\alpha$ has a single justification $J = \{(A\sqcap B)(a) \}$,
	and thus $H = \{\beta = (A\sqcap B)(a) \}$ is the only hitting set. The assertion $\gamma = B(a)$ is weaker than $\beta$ and it
	satisfies $(J\setminus \{\beta\})\cup\{\gamma\} \not\models \alpha$. However, if we define
	$\On' = (\On\setminus \{\beta\})\cup\{\gamma\}$, then $\On'\models \alpha$ still holds.
%
\end{proof}

       A similar example that uses only GCIs is the following, where now we consider a refutable ontology
       $\On = \Onr = \{C \sqsubseteq A\sqcap B, B\sqsubseteq A \}$ and we
       assume that $\alpha$ is the consequence $C \sqsubseteq A$.
       Then $\alpha$ has a single justification $J = \{C\sqsubseteq A\sqcap B\}$
       and thus $H = \{\beta = C\sqsubseteq A\sqcap B\}$ is the only hitting set. The GCI $\gamma = C\sqsubseteq B$ is a weaker than $\beta$ and it
       satisfies $(J\setminus \{\beta\})\cup\{\gamma\} \not\models \alpha$. However, if we define
       $\On' = (\On\setminus \{\beta\})\cup\{\gamma\}$, then $\On'\models \alpha$.

These examples show that applying the gentle repair approach only once may not lead to a repair.
For this reason, we need to \emph{iterate this approach}, i.e., if the resulting ontology $\Ons\cup\On'$ still has $\alpha$
as a consequence, we again compute all justifications and a hitting set for them, and then replace the elements
of the hitting set with weaker axioms as described above. This is iterated until a repair is reached.
We can show that this iteration indeed always terminates after finitely many steps with a repair.

\begin{theorem}\label{always:term:prop}
	Let $\On^{(0)} = \Ons^{(0)}\cup \Onr^{(0)}$ be a finite ontology and $\alpha$ an axiom such that 
	$\On^{(0)}\models\alpha$ and $\Ons^{(0)}\not\models\alpha$. Applied to $\On^{(0)}$ and $\alpha$,
	the iterative algorithm described above stops after a finite number of iterations that is at most exponential in the
	cardinality of $\Onr^{(0)}$, and yields as output an ontology 
	that is a repair of $\Ons^{(0)}$ w.r.t.\ the consequence $\alpha$. 
\end{theorem}

\begin{proof}
	%
	Assume that $\Onr^{(0)}$ contains $n$ axioms, and that there is an infinite run $R$ of the algorithm on input $\On^{(0)}$ and $\alpha$. 
	Take a bijection $\ell_0$ between $\Onr^{(0)}$ and $\{1,\ldots,n\}$ that assigns unique labels to axioms. Whenever 
	we weaken an axiom during a step of the run, the new weaker axiom inherits the label of the original 
	axiom. Thus, we have bijections $\ell_i : \Onr^{(i)}\rightarrow \{1,\ldots,n\}$ for all ontologies $\Onr^{(i)}$
	considered during the run $R$ of the algorithm. For $i\geq 0$ we define
	$$
	\begin{array}{lll}
	S_i : = & \{ K\subseteq \{1,\ldots,n \} \mid \\ [.3em]
	& \Ons\cup\{\beta\in\Onr^{(i)} \mid \ell_i(\beta)\in K\}\models\alpha\},
	\end{array}
	$$
	i.e., $S_i$ contains all sets of indices such that the corresponding subset of $\Onr^{(i)}$ together with $\Ons$
	has the consequence $\alpha$.
	
	We claim that $S_{i+1}\subset S_i$. Note that $S_{i+1}\subseteq S_i$ is an immediate consequence of the fact that
	$\ell_i(\gamma) = j = \ell_{i+1}(\gamma')$ implies that $\gamma = \gamma'$ or $\gamma'$ is weaker than $\gamma$.
	Thus, it remains to show that the inclusion is strict. This follows from the following observations.
	Since the algorithm does not terminate with the ontology $\Onr^{(i)}$, 
	we still have $\Ons\cup\Onr^{(i)}\models\alpha$, and thus there is at least one justification $\emptyset\subset J\subseteq \Onr^{(i)}$.
	Consequently, the hitting set $H$ used in this step of the algorithm contains an element $\beta$ of $\Onr^{(i)}$. When going
	from $\Onr^{(i)}$ to $\Onr^{(i+1)}$, $\beta$ is replaced by a weaker axiom $\beta'$ such that 
	$\Ons\cup(J\setminus \{\beta\})\cup\{\beta'\}\not\models\alpha$.
	But then the set $\{\ell(\gamma) \mid \gamma\in J\}$ belongs to $S_i$, but not to $S_{i+1}$.
	
	Since $S_0$ contains only exponentially many sets, the strict inclusion $S_{i+1}\subset S_i$ can happen only exponentially often, which
	contradicts our assumption that there is an infinite run $R$ of the algorithm on input $\On^{(0)}$ and $\alpha$.
	This shows termination after exponentially many steps. However, if the algorithm terminates with output 
	$\Onr^{(i)}$, then $\Ons\cup\Onr^{(i)}\not\models\alpha$.
	In fact, otherwise, there would be a possibility to weaken $\Onr^{(i)}$ into $\Onr^{(i+1)}$ since it would always be possible to replace the
	elements of a hitting set by tautologies, i.e., perform a classical repair.
\end{proof}

When computing a classical repair, considering all justifications and then removing a minimal hitting set of these justifications guarantees
that one immediately obtains a repair. We have seen in the proof of Lemma~\ref{iteration:needed:lem} that with our gentle repair approach
this need not be the case. Nevertheless, we were able to show that, after a finite number of iterations of the approach, we obtain a repair.
The proof of termination actually shows that for this it is sufficient to weaken only one axiom of one justification such that
the resulting set is no longer a justification. This motivates the following modification of our approach: 

\paragraph{Modified gentle repair algorithm:}
compute one justification $J$ for $\alpha$ in $\On$ and choose an axiom $\beta\in J$.
Replace $\beta$ by a weaker axiom $\gamma$ such that
\begin{equation}\label{modified:repair:cond}
\Ons\cup(J\setminus \{\beta\})\cup\{\gamma\} \not\models \alpha.
\end{equation}
Clearly, one needs to iterate this approach, but it is easy to see that the termination argument used in the proof of
Proposition~\ref{always:term:prop} also applies here. 

\begin{corollary}
	Let $\On^{(0)} = \Ons^{(0)}\cup \Onr^{(0)}$ be a finite ontology and $\alpha$ an axiom such that
	$\On^{(0)}\models\alpha$ and $\Ons^{(0)}\not\models\alpha$. Applied to $\On^{(0)}$ and $\alpha$,
	the modified iterative algorithm stops after a finite number of iterations that is at most exponential in the
	cardinality of $\Onr^{(0)}$, and yields as output an ontology $\widehat{\On}_s$
	that is a repair of $\Ons^{(0)}$ w.r.t.\ $\alpha$.
\end{corollary}

An important advantage of this modified approach is that the complexity of a single iteration step may decrease
considerably. For example, for the DL $\el$, a single justification can be computed in polynomial time,
while computing all justifications may take exponential time \cite{BaPS07}. In addition, to compute a minimal hitting
set one needs to solve an NP-complete problem \cite{GaJo79} whereas choosing one axiom from a single justification is easy.
However, as usual, there is no free lunch: we can show that the modified gentle repair algorithm may indeed need
exponentially many iteration steps.\footnote{%
It is not clear yet whether this is also the case for the unmodified gentle repair algorithm.
}


\begin{proposition}
There is a sequence of \el ontologies $\On^{(n)}  = \Ons^{(n)}\cup\Onr^{(n)}$ with $\Ons^{(n)} = \emptyset$
and an \el axiom $\alpha$ such that the modified gentle repair algorithm applied to $\On^{(n)}$ and $\alpha$
has a run with exponentially many iterations in the size of $\On^{(n)}$.
\end{proposition}

\begin{proof}
	For $n \geq 1$, consider the set of concept names $I^{(n)}=\{P_i,Q_i\mid 1\le i\le n\}$, and define 
	$\On^{(n)} :=  \Onr^{(n)} := \T_1^{(n)} \cup \T_2^{(n)}$, where 
	$$
	\begin{array}{@{}l@{\ }l@{}}
	\T_1^{(n)}:= & \{A \sqsubseteq \exists r.\bigsqcap I^{(n)}, \ \exists r.(P_n\sqcap Q_n)\sqsubseteq B\} \ \cup  \\ 
	& \{ P_{i}\sqcap Q_{i}\sqsubseteq P_{i+1}, \ P_{i}\sqcap Q_{i}\sqsubseteq Q_{i+1} \mid 1\le i< n\}, \\ [.3em]
	\T_2^{(n)} := & \{ \exists r.(X\sqcap Y) \sqsubseteq D_{XY}, \ D_{XY} \sqcap X \sqsubseteq Y \mid \\ 
	& X\in\{P_i,Q_i\}, Y\in\{P_{i+1},Q_{i+1}\},1\le i<n\} \ \cup \\ 
	& \{\exists r. P_1\sqsubseteq P_1, \ \exists r. Q_1 \sqsubseteq Q_1, \
	P_n \sqsubseteq B, \ Q_n \sqsubseteq B\}.
	\end{array}
	$$
	It is easy to see that the size of $\On^{(n)}$ is polynomial in $n$ and that
        $\On^{(n)}\models A\sqsubseteq B$. Suppose that we want to get rid of this consequence using the
	modified gentle repair approach. First, we can find the justification
	\[
	\{ A \sqsubseteq \exists r.\bigsqcap I^{(n)},  \ \exists r.(P_n\sqcap Q_n)\sqsubseteq B\}.
	\]
	We repair it by weakening the first axiom to 
	$$\gamma:=A\sqsubseteq \exists r.\bigsqcap(I^{(n)}\setminus\{P_n\})\ \ \sqcap\ \ \exists r.\bigsqcap(I^{(n)}\setminus\{Q_n\}).$$ 
	At this point, we can find a justification that uses $\gamma$ and $P_{n-1}\sqcap Q_{n-1}\sqsubseteq P_n$. We
	further weaken $\gamma$ to 
	$$
	\begin{array}{ll}
	A\sqsubseteq & \exists r.\bigsqcap(I^{(n)}\setminus\{P_n,P_{n-1}\}) \ \sqcap \\ [.3em]
	& \exists r.\bigsqcap(I^{(n)}\setminus\{P_n,Q_{n-1}\}) 
	\sqcap\exists r.\bigsqcap(I^{(n)}\setminus\{Q_n\}).
	\end{array}
	$$
	Repeating this approach, 
        after $2n$ weakenings
	we have only changed the first axiom, weakening it to the axiom
	\begin{align}
	A \sqsubseteq \bigsqcap_{X_i\in\{P_i,Q_i\},1\le i\le n} \exists r.(X_1\sqcap\cdots \sqcap X_n),
	\label{axiom:weak}
	\end{align}
	whose right-hand side is a conjunction with $2^n$ conjuncts, each of them representing a possible choice of $P_i$ or $Q_i$ at
	every location $i,1\le i\le n$.
	
	So far, we have just considered axioms from $\T_1^{(n)}$. Taking also 
        axioms from $\T_2^{(n)}$ into account, we obtain for every conjunct 
        $\exists r.(X_1\sqcap\cdots \sqcap X_n)$ in axiom \eqref{axiom:weak}
	a justification for 
	$A\sqsubseteq B$ that consists of \eqref{axiom:weak} and the axioms
	$$
	\begin{array}{l}
	\{\, \exists r.X_1\sqsubseteq X_1, \  X_n \sqsubseteq B \, \} \ \cup \\ [.3em]
	\{ \, \exists r.(X_i\sqcap X_{i+1}) \sqsubseteq D_{X_iX_{i+1}},  
	\ D_{X_iX_{i+1}} \sqcap X_i \sqsubseteq X_{i+1} \mid 1\le i<n \, \}. 
	\end{array}
	$$
	This justification can be removed by weakening \eqref{axiom:weak} further by deleting one concept name appearing in the 
	conjunct. The justifications for other conjuncts are not influenced by this modification.
        Thus, we can repeat this for each of the exponentially many conjuncts, which shows that overall we have 
        exponentially many iterations of the modified gentle repair algorithm in this run.
\end{proof}

\subsection{Weakening Relations}

In order to obtain better bounds on the number of iterations of our algorithms, we restrict the way in which
axioms can be weakened. Before introducing concrete approaches for how to do this for \el axioms in the next section, 
we investigate such restricted weakening relations in a more abstract setting.

\begin{definition}
Given a pre-order $\succ$ (i.e., an irreflexive and transitive binary relation) on axioms, we say that it 
\begin{itemize}
\item
  is a \emph{weakening relation} if $\beta\succ\gamma$ implies that $\Con(\{\gamma\})\subset\Con(\{\beta\})$; 
\item
 is \emph{bounded (linear, polynomial)} if, for every axiom $\alpha$, there is a (linear, polynomial) bound $b(\alpha)$
 on the length of all $\succ$-chains issuing from $\alpha$;
\item
 is \emph{complete} if, for any axiom $\beta$ that is not a tautology, there
  is a tautology $\gamma$ such that $\beta\succ\gamma$.
\end{itemize}
\end{definition}
If we use a linear (polynomial) and complete weakening relation, 
then termination with a repair is guaranteed after a linear (polynomial) number of iterations.

\begin{proposition}\label{linterm:prop2}
	Let $\succ$ be a linear (polynomial) and complete weakening relation.
	If in the above (modified) gentle repair algorithm we have $\beta \succ \gamma$ whenever $\beta$ is replaced by $\gamma$,
	then the algorithm stops after a linear (polynomial) number of iterations and yields as output an ontology 
	that is a repair of $\On = \Ons\cup\Onr$ w.r.t.\ the consequence $\alpha$.
\end{proposition}

\begin{proof}
	For every axiom $\beta$ in $\Onr$ we consider the length of the longest $\succ$-chain issuing from it, and
	then sum up these numbers over all axioms in $\Onr$. The resulting number is linearly (polynomially) bounded by the size of the ontology
	(assuming that this size is given as sum of the sizes of all its axioms). Let us call this number the chain-size of the ontology.
	Obviously, if $\beta$ is replaced by $\beta'$ with $\beta \succ \beta'$, then the length of the longest $\succ$-chain issuing from
	$\beta'$ is smaller than the length of the longest $\succ$-chain issuing from $\beta$.
	Consequently, if $\Onr^{(i+1)}$ is obtained from $\Onr^{(i)}$ in the $i$-th iteration of the algorithm, then the chain-size of 
	$\Onr^{(i)}$  is strictly larger than the chain-size of $\Onr^{(i+1)}$. This implies that there can be only linearly (polynomially) 
	many iterations.
	
	Consider a terminating run of the algorithm that has produced the sequence of ontologies
	$\Onr = \Onr^{(0)}, \Onr^{(1)}, \ldots, \Onr^{(n)}$. 
	Then we have $$\Con(\Ons\cup \Onr) \supseteq \Con(\Ons\cup \Onr^{(1)}) \supseteq \ldots \supseteq \Con(\Ons\cup \Onr^{(n)})$$
	since $\succ$ is a weakening relation.
	If the algorithm has terminated due to the fact that $\alpha\not\in \Con(\Ons\cup \Onr^{(n)})$, then $\Onr^{(n)}$ is a repair of
	$\On$ w.r.t.\ $\alpha$. Otherwise, the only reason for termination could be that, although $\alpha\in \Con(\Ons\cup \Onr^{(n)})$,
	the algorithm cannot generate a new ontology $\Onr^{(n+1)}$. In the unmodified gentle repair approach this means that there is an
	axiom $\beta$ in the hitting set $H$ such that there is no axiom $\gamma$ with $\beta\succ \gamma$ such that (\ref{replace:prop})
        is satisfied.
	%
	However, using a tautology as the axiom $\gamma$ actually allows us to satisfy
	the condition $(\ref{replace:prop})$. Thus, completeness of $\succ$ implies that this reason for termination
	without success cannot occur. An analogous argument can be used for the modified gentle repair approach.
\end{proof}

When describing our (modified) gentle repair algorithm, we have said that the chosen axiom $\beta$ needs to be replaced by a weaker axiom $\gamma$
such that (\ref{replace:prop}) or (\ref{modified:repair:cond}) holds. But we have not said how such an axiom $\gamma$ can be found. This of course
depends on which ontology language and which weakening relation is used. In the abstract setting of this section, we assume that an 
``oracle'' provides us with a weaker axiom.

\begin{definition}
Let $\succ$ be a weakening relation. An
\emph{oracle for $\succ$} is a computable function $W$ that, given an axiom $\beta$ that is not $\succ$-minimal,
provides us with an axiom $W(\beta)$ such that $\beta\succ W(\beta)$. For $\succ$-minimal axioms $\beta$ we assume that $W(\beta) = \beta$.
\end{definition}

If the weakening relation is complete and \emph{well-founded} 
(i.e., there are no infinite descending $\succ$-chains $\beta_1\succ\beta_2\succ\beta_2\succ\cdots$),
we can effectively find an axiom $\gamma$ such that (\ref{replace:prop}) or (\ref{modified:repair:cond}) holds. We show this formally only for
(\ref{modified:repair:cond}), but condition (\ref{replace:prop}) can be treated similarly.

%
%

\begin{lemma}\label{lem:just}
Assume that $J$ is a justification for the consequence $\alpha$, and $\beta \in J$. If $\succ$ is a well-founded and complete weakening relation and
$W$ is an oracle for $\succ$, then there is an $n\geq 1$ such that (\ref{modified:repair:cond}) holds for $\gamma = W^n(\beta)$.
If $\succ$ is additionally linear (polynomial), then $n$ is linear (polynomial) in the size of $\beta$.
\end{lemma}

\begin{proof}
Well-foundedness implies that the $\succ$-chain $\beta \succ W(\beta)\succ W(W(\beta)) \succ \ldots$ is finite, and thus there is an
$n$ such that $W^{n+1}(\beta) = W^n(\beta)$, i.e., $W^n(\beta)$ is $\succ$-minimal.  Since $\succ$ is complete, this implies that
$W^n(\beta)$ is a tautology. Minimality of the justification $J$ then yields $\Ons\cup (J \setminus \{\beta\}) \cup \{W^n(\beta) \} \not\models \alpha$.
Linearity (polynomiality) of $\succ$ ensures that the length of the $\succ$-chain $\beta \succ W(\beta)\succ W(W(\beta)) \succ \ldots$
is linearly (polynomially) bounded by the size of $\beta$.
\end{proof}

Thus, to find an axiom $\gamma$ satisfying (\ref{replace:prop}) or (\ref{modified:repair:cond}), we iteratively apply $W$ to $\beta$ until
an axiom satisfying the required property is found. The proof of Lemma~\ref{lem:just} shows that at the latest this is the case when a tautology is reached,
but of course the property may already be satisfied before that by a non-tautological axiom $W^i(\beta)$.

In order to weaken axioms as gently as possible, $W$ should realize small weakening steps. The smallest such step is one where there is no step 
in between.

\begin{definition}
Let $\succ$ be a pre-order. The \emph{one-step relation}\footnote{%
This is sometimes also called the transitive reduction of $\succ$.
} 
induced by $\succ$ is defined as
$$
{\succ_1} := \{(\beta,\gamma) \in {\succ} \mid\ \mbox{there is no $\delta$ such that}\ \beta\succ\delta\succ\gamma \}.
$$
We say that $\succ_1$ \emph{covers} $\succ$ if its transitive closure is again $\succ$, i.e., ${\succ_1^+} = {\succ}$.
In this case we also say that $\succ$ is \emph{one-step generated}.
\end{definition}
If $\succ$ is one-step generated, then every weaker element can be reached by a finite sequence of one-step weakenings,
i.e., if $\beta\succ\gamma$, then there are finitely many elements $\delta_0,\ldots,\delta_n$ ($n\geq 1$) such that 
$\beta = \delta_0\succ_1 \delta_1\succ_1 \ldots\succ_1\delta_n=\gamma$. This leads us to the following characterization of
pre-orders that are \emph{not} one-step generated.

\begin{lemma}\label{not:one:step:gen:lem}
The pre-order $\succ$ is \emph{not} one-step generated iff there exist two comparable elements $\beta\succ\gamma$
such that every finite chain $\beta = \delta_0\succ \delta_1\succ \ldots\succ\delta_n=\gamma$ can be \emph{refined} in the sense
that there is an $i, 0\leq i < n$, and an element $\delta$ such that $\delta_i\succ \delta\succ\delta_{i+1}$.
\end{lemma}

If $\beta\succ\gamma$ are such that any finite chain between them can be refined, then obviously there cannot be an
upper bound on the length of the chains issuing from $\beta$. Thus, Lemma~\ref{not:one:step:gen:lem} implies the following
result.

\begin{proposition}\label{bounded:then:one:step:gen}
If $\succ$ is bounded, then it is one-step generated.
\end{proposition}


The following example shows that well-founded pre-orders need not be one-step generated.

\begin{example}\label{one:step:counter:ex}
Consider the pre-order $\succ$ on the set 
$$
P := \{\beta\}\cup\{\delta_i \mid i \geq 0\},
$$
where $\beta\succ\delta_i$ for all $i\geq 0$, and $\delta_i\succ\delta_j$ iff $i>j$.
It is easy to see that $\succ$ is well-founded and that
$
{\succ_1} = \{(\delta_{i+1},\delta_i) \mid i\geq 0\}.
$
Consequently, ${\succ_1}^+$ contains none of the tuples $(\beta,\delta_i)$ for $i\geq 0$, which shows that
${\succ_1}$ does not cover $\succ$. In particular, any finite chain between $\beta$ and $\delta_i$ can be refined.

Interestingly, if we add elements $\gamma_i$ ($i\geq 0$) with $\beta\succ\gamma_i\succ\delta_i$ to this
pre-order, then it becomes one-step generated.
\end{example}

One-step generated weakening relations allow us to find maximally strong weakenings satisfying 
(\ref{replace:prop}) or (\ref{modified:repair:cond}). Again, we consider only condition
(\ref{modified:repair:cond}), but all definitions and results can be adapted to deal with (\ref{replace:prop}) as well.

\begin{definition}
Let $J$ be a justification for the consequence $\alpha$, and $\beta \in J$. We say that 
$\gamma$ is a \emph{maximally strong weakening} of $\beta$ in $J$ if
$\Ons\cup (J \setminus \{\beta\}) \cup \{\gamma\} \not\models \alpha$, but 
$\Ons\cup (J \setminus \{\beta\}) \cup \{\delta\} \models \alpha$ for all $\delta$ with $\beta\succ\delta\succ\gamma$.
\end{definition}

In general, maximally strong weakenings need not exist. As an example, assume that the pre-order introduced
in Example~\ref{one:step:counter:ex} (without the added axioms $\gamma_i$) is a weakening relation on axioms, and assume that $J = \{\beta\}$
and that none of the axioms $\delta_i$ have the consequence.
Obviously, in this situation there is no maximally strong weakening of $\alpha$ in $J$.
 
Next, we introduce conditions under which maximally strong weakenings always exist, and can also be computed. We say that the one-step generated weakening relation
$\succ$ is \emph{effectively finitely branching} if for every axiom $\beta$ the set 
$
\{\gamma \mid \beta\succ_1\gamma\}
$
is finite and can effectively be computed.

\begin{proposition}\label{max:strong:exist:prop}
Let $\succ$ be a well-founded, one-step generated, and effectively finitely branching  weakening relation 
and assume that the consequence relation $\models$ is decidable. Then all maximally strong weakenings
of an axiom in a justification can effectively be computed.
\end{proposition}

\begin{proof}
Let $J$ be a justification for the consequence $\alpha$, and $\beta \in J$.
Since $\succ$ is well-founded, one-step generated, and finitely branching, K\"onig's Lemma implies that there are
only finitely many $\gamma$ such that $\beta\succ\gamma$, and all these $\gamma$ can be reached by following
$\succ_1$. Thus, by a breadth-first search, we can compute the set of all $\gamma$ such that there is a path
$\beta\succ_1 \delta_1\succ_1 \ldots\succ_1\delta_n\succ_1\gamma$ with 
$\Ons\cup (J \setminus \{\beta\}) \cup \{\gamma\} \not\models \alpha$, but
$\Ons\cup (J \setminus \{\beta\}) \cup \{\delta_i\} \models \alpha$ for all $i, 1\leq i\leq n$.
If this set still contains elements that are comparable w.r.t.\ $\succ$ (i.e., there is a
$\succ_1$-path between them), then we remove the weaker elements. It is easy to see that the remaining
set consists of all maximally strong weakenings of $\beta$ in $J$.
\end{proof}

Note that the additional removal of weaker elements in the above proof is really necessary. In fact, assume that
$\beta\succ_1\delta_1\succ_1\gamma$ and $\beta\succ_1\delta_2\succ_1\gamma$, and that
$\Ons\cup (J \setminus \{\beta\}) \cup \{\gamma\} \not\models \alpha$, 
$\Ons\cup (J \setminus \{\beta\}) \cup \{\delta_1\} \models \alpha$, but
$\Ons\cup (J \setminus \{\beta\}) \cup \{\delta_2\} \not\models \alpha$.
Then both $\delta_2$ and $\gamma$ belong to the set computed in the breadth-first search,
but only $\delta_2$ is a maximally strong weakening (see Example~\ref{ex:onestep-el}, where it is shown
that this situation can really occur when repairing \el ontologies).

In particular, this also means that iterated application of a one-step oracle, i.e., an oracle $W$ satisfying
$\beta\succ_1 W(\beta)$, does not necessarily yield a maximally strong weakening.


\section{Weakening Relations for $\el$ Axioms}
\label{sec:weakening}

In this section, we restrict the attention to ontologies written in $\el$, but some of our approaches and results could also be transferred to
other DLs. We start with observing that weakening relations for $\el$ axioms need not be one-step generated.

\begin{proposition}
	If we define $\beta \succ^g \gamma$ if $\Con(\gamma) \subset \Con(\beta)$, then $\succ^g$ is a weakening relation 
	on $\el$ axioms that is not one-step generated.
\end{proposition}

\begin{proof}
	It is obvious that $\succ^g$ is a weakening relation.\footnote{%
		In fact, it is the greatest one w.r.t.\ set inclusion.
	} 
	To see that it is not one-step generated, consider a
	GCI $\beta$ that is not a tautology and an arbitrary tautology $\gamma$. Then we
	have $\beta \succ \gamma$. Let $\beta=\delta_0\succ^g\delta_1\succ^g\ldots\succ^g\delta_n=\gamma$ be a finite
	chain leading from $\beta$ to $\gamma$. Then $\delta_{n-1}$ must be a GCI that is not a tautology.
	Assume that $\delta_{n-1} = {C\sqsubseteq D}$. Then $\delta := {\exists r.C\sqsubseteq \exists r.D}$
	satisfies $\delta_{n-1}\succ^g\delta\succ^g \gamma$. By Lemma~\ref{not:one:step:gen:lem}, this shows that
	$\succ$ is not one-step generated.
\end{proof}

%
%

Our main idea for obtaining more well-behaved weakening relations is to 
weaken a GCI $C\sqsubseteq D$ by generalizing the right-hand side $D$ and/or by specializing
the left-hand side $C$. Similarly, a concept assertion $D(a)$ can be weakened by generalizing $D$.
For role assertions we can use as weakening an arbitrary tautological axiom,
but will no longer consider them explicitly in the following. 

\begin{proposition}\label{weakening:prop}
	If we define
	$$
	\begin{array}{r@{\ \ \ }c@{\ \ \ }l}
	{C\sqsubseteq D} \succ^s {C'\sqsubseteq D'} &\mbox{if}& C'\subeq C,\ D\subeq D'\ \mbox{and} \
	\{C'\sqsubseteq D'\}\not\models C\sqsubseteq D,\\
	D(a) \succ^s D'(a) &\mbox{if}& D \sub D',
	\end{array}
	$$
	then $\succ^s$ is a complete weakening relation.
\end{proposition}

\begin{proof}
	To prove that $\succ^s$ is a weakening relation we must show that $\beta\succ^s\gamma$ implies $\Con(\{\gamma\})\subset \Con(\{\beta\})$.
	If $C'\subeq C$ and $D\subeq D'$ hold, then it follows that
	$\Con(\{C'\sqsubseteq D'\})\subseteq\Con(\{C\sqsubseteq D\})$
	and $\Con(\{a: D'\})\subseteq\Con(\{a :D\})$. The second inclusion is strict iff $D \sub D'$.
	For the first inclusion to be strict, $C'\sub C$ or $D\sub D'$ is a necessary condition, but
	it is not sufficient. This is why we explicitly require $\{C'\sqsubseteq D'\}\not\models C\sqsubseteq D$,
	which yields strictness of the inclusion. Completeness is trivial due to the availability of all tautologies of the form
	$C\sqsubseteq \top$ and $\top(a)$.
\end{proof}

To see why, e.g., $D\sub D'$ does not imply $\Con(\{C\sqsubseteq D'\})\subset\Con(\{C\sqsubseteq D\})$, 
notice that $A\sqcap \exists r.A\sub \exists r.A$, but
$\Con(\{A \sqsubseteq \exists r.A\})=\Con(\{A\sqsubseteq A\sqcap\exists r.A\})$.


Unfortunately, the weakening relation $\succ^s$ introduced in Proposition~\ref{weakening:prop} is \emph{not well-founded} since left-hand sides can be specialized indefinitely.
For example, we have ${\top\sqsubseteq A}\succ^s {\exists r.\top \sqsubseteq A}\succ^s {\exists r.\exists r.\top \sqsubseteq A}\succ^s\cdots$. To avoid this problem, we
now restrict the attention to sub-relations of $\succ^s$ that only generalize the right-hand sides of GCIs. We will not consider concept assertions, but they can be treated similarly.

\subsection{Generalizing the Right-Hand Sides of GCIs}
\label{subsec: weakeningEL}

We define
$$
{C\sqsubseteq D} \susu {C'\sqsubseteq D'}\ \ \mbox{if}\ \ C' = C\ \mbox{and}\ {C\sqsubseteq D} \succ^s {C'\sqsubseteq D'}. 
$$

\begin{theorem}\label{weakening:rel:el}
	The relation $\susu$ on $\el$ axiom is a well-founded, complete, and one-step generated weakening relation, but it is not polynomial.
\end{theorem}

\begin{proof}
	Proposition~\ref{weakening:prop} implies that $\susu$ is a weakening relation and completeness follows from the fact that ${C\sqsubseteq D}\susu {C\sqsubseteq \top}$ whenever
	${C\sqsubseteq D}$ is not a tautology.
	In $\el$, the inverse subsumption relation is well-founded, i.e., there cannot be an infinite sequence
	$C_0\sub C_1\sub C_2\sub \ldots$ of $\el$ concepts. 
	Looking at the proof of this result given in  \cite{BaMo10}, one sees that it actually shows that $\sub$ is bounded.
	Obviously, this implies that $\susu$ is bounded as well, and thus one-step generated by Proposition~\ref{bounded:then:one:step:gen}.
	
	It remains to show that $\susu$ is not polynomial.
	Let $n\geq 1$ and $N_n := \{A_1,\ldots,A_{2n}\}$ be a set of $2n$ distinct concept names. 
        Then we have
	$$
	\exists r.{\bigsqcap} N_n \sub \bigsqcap_{X\subseteq N_n\wedge |X| = n}\exists r.{\bigsqcap}{X}.
	$$
	Note that the size of $\exists r.{\bigsqcap} N_n$ is linear in $n$, but that the conjunction  on the right-hand side
	of this strict subsumption consists of exponentially many concepts $\exists r.{\bigsqcap}{X}$ that are incomparable
	w.r.t.\ subsumption. Consequently, by removing one conjunct at a time, we can generate an ascending chain
	w.r.t.\ $\sub$ of \el concepts whose length is exponential in $n$.
	Using these concepts as right-hand sides of GCIs with left-hand side $B$ for a concept name $B\not\in N_n$, we
	obtain an exponentially long descending chain w.r.t.\ $\susu$.
\end{proof}

To be able to apply Proposition~\ref{max:strong:exist:prop}, it remains to show that 
$\susu$
is effectively finitely branching. For this purpose, we first investigate the one-step relation $\sub_1$ induced by $\sub$. Given an $\el$ concept
$C$, we want to characterize the set of its \emph{upper neighbors} 
$$
\Upper(C) := \{ D \mid C\sub_1 D\},
$$
and show that it can be computed in polynomial time.

In a first step, we \emph{reduce} the concept $C$ by exhaustively replacing subconcepts of the form $E\sqcap F$ with $E \subeq F$ by $E$ (modulo
associativity and commutativity of $\sqcap$). As shown in \cite{Kues01}, this can be done in polynomial time, and two concepts $C, D$ are equivalent 
(i.e., $C\equiv^\emptyset D$) iff their reduced forms are equal up to associativity and commutativity of $\sqcap$. 
 
\begin{definition}
Given a reduced $\el$ concept $C$, we define the set $U(C)$ by induction on the role depths of $C$.
More precisely, $U(C)$ consists of the concepts $D$ that can be obtained from $C$ as follows:
       \begin{itemize}
               \item
               Remove a concept name $A$ from the top-level conjunction of $C$.
               \item
               Remove an existential restriction $\exists r.E$ from the top-level conjunction of $C$, and replace it by
               the conjunction of all existential restrictions $\exists r.F$ for $F\in U(E)$.
       \end{itemize}
\end{definition}

%
For example, if $C = A\sqcap \exists r.(B_1\sqcap B_2\sqcap B_3)$, then $U(C)$ consists of the two concepts 
$\exists r.(B_1\sqcap B_2\sqcap B_3)$ and $A\sqcap \exists r.(B_1\sqcap B_2)\sqcap\exists r.(B_1\sqcap B_3)\sqcap\exists r.(B_2\sqcap B_3)$.

We want to prove that $\Upper(C) = U(C)$.
Obviously, this shows that $\Upper(C)$ can be computed in time polynomial in the size of $C$. 
But first we we need to show some technical lemmas.

\begin{lemma}\label{in:U:strict:subs}
Let $C$ be reduced and assume that $D\in U(C)$. Then $C\sub D$.
\end{lemma}

\begin{proof}
We prove the lemma by induction on the role depths of $C$. If $D$ is obtained from $C$
by removing a concept name from the top-level conjunction of $C$, then $C\sub D$ is an immediate consequence of Lemma~\ref{subs:char:lem}.

Thus, assume that $D$ is obtained from $C$ by replacing an existential restriction $\exists r.E$ from the top-level conjunction of $C$
with the conjunction of all existential restrictions $\exists r.F$ for $F\in U(E)$. Then induction yields $E\sub F$ for
all $F\in U(E)$. Thus, $C\subeq D$ is an immediate consequence of Lemma~\ref{subs:char:lem}. Now, assume that $D\subeq C$. By Lemma~\ref{subs:char:lem} this implies
that there is an existential restriction $\exists r.D'$ in the top-level conjunction of $D$ such that $D'\subeq E$. Obviously, $D'\not\in U(E)$ since
in that case we would have $E\sub D'$. Thus, $\exists r.D'$ is an existential restriction different from $\exists r.E$ from the top-level conjunction of $C$.
But then $D'\subeq E$ contradicts our assumption that $C$ is reduced. Thus, we have shown $C\sub D$  also in this case.
\end{proof}

\begin{lemma}\label{U:in:between}
Let $C$ be reduced and assume that $C\sub D$. Then there is $D'\in U(C)$ such that $D'\subeq D$.
\end{lemma}

\begin{proof}
Again, we prove the lemma by induction on the role depths of $C$.
Let $C = A_1\sqcap\ldots\sqcap A_k\sqcap\exists r_1.C_1\sqcap\ldots\sqcap\exists r_m.C_m$ and
$D = B_1\sqcap\ldots\sqcap B_\ell\sqcap\exists s_1.D_1\sqcap\ldots\sqcap\exists s_n.D_n$ for concept names $A_1,\ldots,A_k,B_1,\ldots B_\ell$.
Since $C\subeq D$, we know by Lemma~\ref{subs:char:lem} that $\{B_1,\ldots, B_\ell\}\subseteq \{A_1,\ldots, A_k\}$ and that for every
$j, 1\leq j\leq n$, there is $i,1\leq i\leq m$ such that $r_i=s_j$ and $C_i\subeq D_j$.

Strictness of the subsumption relationship $C\sub D$ may be due to the fact that $\{B_1,\ldots, B_\ell\}\subset \{A_1,\ldots, A_k\}$.
In this case, let $A \in \{A_1,\ldots, A_k\}\setminus \{B_1,\ldots, B_\ell\}$, and let $D'$ be obtained from $C$
by removing the concept name $A$ from the top-level conjunction of $C$. Then $D'\in U(C)$, and $D'\subeq D$ holds by Lemma~\ref{subs:char:lem}.

Now assume that $\{B_1,\ldots, B_\ell\} = \{A_1,\ldots, A_k\}$. Then $D\not\subeq C$ implies that there is an $i,1\leq i\leq m$ such that for
all $j, 1\leq j\leq n$ with $r_i=s_j$ we have $D_j\not\subeq C_i$. Let $D'$ be obtained from $C$ by replacing the existential restriction 
$\exists r_i.C_i$ from the top-level conjunction of $C$ with the conjunction of all existential restrictions $\exists r_i.F$ for $F\in U(C_i)$.
Then $D'\in U(C)$ and it remains to prove that $D'\subeq D$. 

To show that the conditions of  Lemma~\ref{subs:char:lem} are satisfied, we consider an existential restriction
$\exists s_j.D_j$ in the top-level conjunction of $D$. Since $C\subeq D$, there is an index $\nu, 1\leq \nu\leq m$ such that
$r_\nu = s_j$ and $C_\nu \subeq D_j$. If $\nu\neq i$, then $\exists r_\nu.C_\nu$ is also a top-level conjunct of $D'$, and thus we are done.
Thus, assume that $\nu = i$. In this case, we know that $D_j\not\subeq C_\nu = C_i$, and thus $C_\nu \sub D_j$. By induction,
there is a concept $F\in U(C_i)$ such that $F\subeq D_j$, and we are again done since $\exists r_i.F$ is a top-level conjunct of $D'$.
\end{proof}

\begin{lemma}\label{U:incomp}
Let $C$ be a reduced $\el$ concept.
If $D$ and $D'$ are different elements of $U(C)$, then $D\not\subeq D'$.
\end{lemma}

\begin{proof}
If $D$ and $D'$ are obtained from $C$ by removing different concept names, then $D\subeq D'$ obviously can\emph{not} hold by Lemma~\ref{subs:char:lem}.
The same is true if $D$ is obtained by removing a concept name and $D'$ is obtained by replacing an existential restriction.

Assume that $D$ is obtained from $C$ by replacing an existential restriction $\exists r.E$ 
with the conjunction of the existential restrictions $\exists r.F$ for $F\in U(E)$.
Since $D, D'$ are different elements of $U(C)$, $\exists r.E$ still belongs to the top-level conjunction of $D'$. Now, $D\subeq D'$ implies that 
there is an existential restriction $\exists r.E'$ in the top-level conjunction of $D$ such that $E'\subeq E$. If $\exists r.E'$ is an original
conjunct in the top-level conjunction of $C$, this contradicts our assumption that $C$ is reduced. Otherwise, $\exists r.E'$ must be such that
$E'\in U(G)$ for an existential restriction $\exists r.G$ different from $\exists r.E$ in the top-level conjunction of $C$. But then
$G\sub E'\subeq E$, which again contradicts our assumption that $C$ is reduced.
\end{proof}

\begin{proposition}
Let $C$ be a reduced $\el$ concept.  Then up to equivalence we have $\Upper(C) = U(C)$.
In particular, this implies that the cardinality of $\Upper(C)$ is polynomial in the size of $C$ and that this set can be computed
in polynomial time in the size of $C$.
\end{proposition}

\begin{proof}
First, assume that $D\in \Upper(C)$, i.e., $C\sub_1 D$. Then Lemma~\ref{U:in:between} implies that there is $D'\in U(C)$ such 
that $D'\subeq D$. But then $C\sub D' \subeq D$ and $C\sub_1 D$ imply $D'\equiv^\emptyset D$, and thus $D$ is equivalent to an element of $U(C)$.

Conversely, assume that $D\in U(C)$. Then Lemma~\ref{in:U:strict:subs} yields $C\sub D$. To show that $C\sub_1 D$, assume to the contrary that
there is a concept $D'$ such that $C\sub D'\sub D$. Then Lemma~\ref{U:in:between} yields the existence of a concept $D''\in U(C)$ such that
$C\sub D''\subeq D'\sub D$. But then $D$ and $D''$ are two different elements of $U(C)$ that are comparable w.r.t.\ $\sub$, which contradicts Lemma~\ref{U:incomp}.

The polynomiality results for $U(C)$ can easily be shown by induction on the role depth of $C$.
\end{proof}

Unfortunately, this result does not transfer immediately from concept subsumption to axiom weakening. In fact, as we have seen before,
strict subsumption need not produce a weaker axiom (see the remark below Proposition~\ref{weakening:prop}). Thus, to find all GCIs
$C\sqsubseteq D'$ with ${C\sqsubseteq D}\susu_1 {C\sqsubseteq D'}$, it is not sufficient to consider only concepts $D'$ with
$D\sub_1 D'$. In case $C\sqsubseteq D'$ is equivalent to $C\sqsubseteq D$, we need to consider upper neighbors of $D'$, etc.

\begin{proposition}
\label{prop:onestep:efb}
	The one-step relation $\susu_1$ induced by $\susu$
	is effectively finitely branching.
\end{proposition}

\begin{proof}
	Since $\sub$ is one-step generated, finitely branching, and well-founded, for a given concept $D$, there are only
	finitely many concepts $D'$ such that $D\sub D'$. Thus, a breadth first search along $\sub_1$ can be
	used to compute all concepts $D'$ such that there is a path $D\sub_1 D_1\sub_1\ldots D_n\sub_1 D'$
	where $C\sqsubseteq D$ is equivalent to $C\sqsubseteq D_i$ for $i=1,\ldots, n$, and
	${C\sqsubseteq D}\susu{C\sqsubseteq D'}$. Since $\sub$ is one-step generated, it is easy to see that all axioms $\gamma$ with
	${C\sqsubseteq D}\susu_1\gamma$ can be obtained this way. However, the computed set of axioms may contain
	elements that are not one-step successors of $C\sqsubseteq D$. Thus, in a final step, we remove all axioms
	that are weaker than some axiom in the set. 
\end{proof}

\begin{figure}
	\centering
	\begin{tikzpicture}[x=7.5mm,y=7.5mm]
	\begin{scope}[every node/.style={draw,rectangle,rounded corners,inner sep=0.3em}]
	\node[double] (0) at ( 0, 0) {$\top\sqsubseteq A\sqcap\exists r.A$}    ;
	\node         (1) at (-2,-2) {$\top\sqsubseteq A\sqcap\exists r.\top$} ;
	\node[double] (2) at ( 2,-2) {$\top\sqsubseteq\exists r.A$}            ;
	\node[double] (3) at ( 0,-4) {$\top\sqsubseteq\exists r.\top$}         ;
	\end{scope}
	\path[-{Classical TikZ Rightarrow[length=1.618mm]},shorten >=1.618mm]
	(0) edge node[below,rotate=-135] {$\models$} node[below,rotate=45]  {$\models$}     (1)
	(0) edge node[above,rotate=-45]  {$\models$} node[above,rotate=135] {$\not\models$} (2)
	(1) edge node[below,rotate=-45]  {$\models$} node[below,rotate=135] {$\not\models$} (3)
	(2) edge node[above,rotate=-135] {$\models$} node[above,rotate=45]  {$\not\models$} (3)
	;
	\end{tikzpicture}
	\caption{One-step weakening}
	\label{fig:gentle-weakening-with-neighbors}
\end{figure}
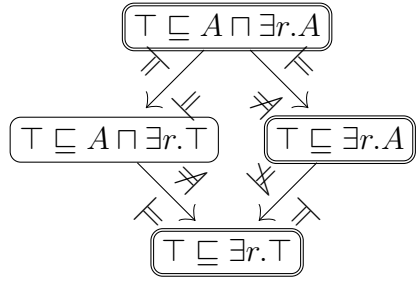

\begin{example}\label{ex:onestep-el}
	To see that the final step of removing axioms in the proof of Proposition~\ref{prop:onestep:efb} is needed, 
	consider the axiom $\beta = \top\sqsubseteq A\sqcap\exists r.A$ in Fig.~\ref{fig:gentle-weakening-with-neighbors}.
	The right-hand side $A\sqcap\exists r.A$ has two upper neighbors, namely $\exists r.A$ and $A\sqcap\exists r.\top$.
	The first yields the axiom $\top\sqsubseteq \exists r.A$, which satisfies 
	${\top\sqsubseteq A\sqcap\exists r.A}\susu_1{\top\sqsubseteq \exists r.A}$.
	The second yields the axiom $\top\sqsubseteq A\sqcap\exists r.\top$, which is equivalent to $\beta$.
	Thus, the only upper neighbor $\top\sqsubseteq\exists r.\top$ is considered, but this concept yields
	an axiom that is actually weaker than $\top\sqsubseteq \exists r.A$, and thus needs to be removed.
	
	A similar, but simpler example can be used to show
	that the additional removal of weaker elements in the proof of Proposition~\ref{max:strong:exist:prop} is needed.
	Let $\alpha$ be the consequence $\top \sqsubseteq A$, $J = \{\beta\}$ for $\beta := \top\sqsubseteq A\sqcap B$, 
	$\delta_1 := \top \sqsubseteq A$, $\delta_2 := \top\sqsubseteq B$, and $\gamma := \top\sqsubseteq \top$.
	Then we have exactly the situation described below the proof of Proposition~\ref{max:strong:exist:prop}, with
	$\susu$ as the employed weakening relation.
\end{example}

\begin{corollary}
	All maximally strong weakenings w.r.t.\ $\susu$ of an axiom in a justification can effectively be computed.
\end{corollary}

\begin{proof}
	By Proposition~\ref{max:strong:exist:prop}, this is an immediate consequence of the fact that $\susu$
	is well-founded, one-step generated, and effectively finitely branching.
\end{proof}

The algorithm for computing maximally strong weakenings described in the proof of Proposition~\ref{max:strong:exist:prop}
has non-elementary complexity for $\susu$. In fact, the bound for the depth of the tree that must be searched grows by one
exponential for every increase in the role-depth of the concept on the right-hand side.
It is not clear how to obtain an algorithm with a better complexity. Example~\ref{exp:max:strong:ex} below yields an
exponential lower-bound, which still leaves a huge gap. We can also show that even deciding whether a given axiom
is a maximally strong weakening w.r.t.\ $\susu$ is coNP-hard.

Before we can prove this hardness result, we must introduce the coNP-complete problem that will be used in
our proof by reduction. A \emph{monotone Boolean formula} $\varphi$ is built from propositional variables using
the connectives conjunction ($\wedge$) and disjunction ($\vee$) only. If $V$ is the set of propositional variables occurring in $\varphi$,
then propositional valuations can be seen as subsets $W$ of $V$. Since $\varphi$ is monotone, the valuation $V$ clearly
satisfies $\varphi$, and the valuation $\emptyset$ falsifies $\varphi$. We are now interested in \emph{maximal} valuations falsifying
$\varphi$, where valuations are compared using set inclusion.

\begin{definition}
The \emph{all-maximal-valuations} problem receives as input 
\begin{itemize}
\item a monotone Boolean formula $\varphi$ with propositional variables $V$, and
\item a set $\Vmc$ of maximal valuations falsifying $\varphi$.
\end{itemize}
The question is then whether $\Vmc$ is the set of \emph{all} maximal valuations falsifying $\varphi$.
\end{definition}

As shown in \cite{Pena09} (Lemma~6.13), the all-maximal-valuations problem is coNP-complete.


\begin{proposition}
The problem of
deciding whether a given \el GCI $C\sqsubseteq D'$ is a maximally strong weakening of the \el GCI $C\sqsubseteq D$ w.r.t.\
$\susu$ is coNP-hard.
\end{proposition}

\begin{proof}
Given an instance $\varphi,\Vmc$ of the all-maximal-valuations problem, we construct an instance of our
problem as follows. For every subformula $\psi$ of $\varphi$, we introduce a new concept name $B_\psi$.
If $\psi$ is not a propositional variable, we define the TBox:
\[
\Tmc_\psi :=
        \begin{cases}
                \{ B_{\psi_1}\sqcap B_{\psi_2}\sqsubseteq B_\psi\} & \psi=\psi_1\land\psi_2, \\
                \{ B_{\psi_1}\sqsubseteq B_\psi, B_{\psi_2}\sqsubseteq B_\psi\} & \psi=\psi_1\lor\psi_2.
        \end{cases}
\]
Let $V$ be the set of all propositional variables appearing in $\varphi$, and let $csub(\varphi)$ be the set
of all subformulas of $\varphi$ that are not in $V$.

We construct the ontology that has only one refutable axiom
$$A\sqsubseteq \exists r.\bigsqcap\{B_p\mid p\in V\},$$
and as static part the ontology
\[
\Tmc_s = \bigcup_{\psi\in csub(\varphi)}\Tmc_\psi \cup \{ \exists r. B_\varphi\sqsubseteq C\}.
\]
Clearly, the refutable axiom is a justification for $A\sqsubseteq C$.

Given a set \Wmc of valuations, define the concept
$$
X_\Wmc := \bigsqcap_{W\in\Wmc}\exists r.\bigsqcap\{B_p\mid p\in W\}.
$$
It follows that
$\{A\sqsubseteq X_\Wmc\}\cup\Tmc_s\not\models A\sqsubseteq C$
iff no valuation in \Wmc satisfies $\varphi$.

We \emph{claim} that \Vmc is the set of all maximal valuations not satisfying $\varphi$ iff $A\sqsubseteq X_\Vmc$  is a maximally strong
weakening of $A\sqsubseteq\exists r.\bigsqcap\{B_p\mid p\in V\}$.

First, assume that \Vmc is the set of all maximal valuations not satisfying $\varphi$. 
Then $\{A\sqsubseteq X_\Vmc\}\cup\Tmc_s\not\models A\sqsubseteq C$ and clearly 
${A\sqsubseteq \exists r.\bigsqcap\{B_p\mid p\in V\}}\susu{A\sqsubseteq X_\Vmc}$. If $A\sqsubseteq X_\Vmc$ is not maximally strong,
then there is a concept $E$ such that $\exists r.\bigsqcap\{B_p\mid p\in V\}\sub E \sub X_\Vmc$
and $\{A\sqsubseteq E\}\cup\Tmc_s\not\models A\sqsubseteq C$. The strict subsumption relationships imply the $E$ contains
a top-level conjunct $\exists r.\bigsqcap\{B_p\mid p\in U\}$ for a set $U\subseteq V$ such that $U$ is incomparable w.r.t.\
set inclusion with all the sets in $\Vmc$. Since \Vmc is the set of all maximal valuations not satisfying $\varphi$, this implies
that $U$ satisfies $\varphi$. Consequently, $\{A\sqsubseteq E\}\cup\Tmc_s\models A\sqsubseteq C$, which yields a contradiction to our
assumption that $A\sqsubseteq X_\Vmc$ is not maximally strong.

Conversely, assume that \Vmc is not the set of all maximal valuations not satisfying $\varphi$, i.e., there
is a maximal valuation $U$ not satisfying $\varphi$ such that $U\not\in \Vmc$. This implies that $U$ is incomparable w.r.t.\ inclusion
with any of the elements of $\Vmc$, and thus $\exists r.\bigsqcap\{B_p\mid p\in V\}\sub X_{\Vmc\cup\{U\}} \sub X_\Vmc$. 
In addition, we know that $\{A\sqsubseteq X_{\Vmc\cup\{U\}}\}\cup\Tmc_s\not\models A\sqsubseteq C$, which shows
that $A\sqsubseteq X_\Vmc$ is not maximally strong.
\end{proof}



\subsection{Syntactic Generalization}

In order to obtain a weakening relation that has better algorithmic properties than $\susu$, we
consider a syntactic approach for generalizing $\el$ concepts. Basically, the concept $D$ is a syntactic
generalization of the concept $C$ if $D$ can be obtained from $C$ by removing occurrences of subconcepts.
To ensure that such a removal really generalizes the concept, we work here with reduced concepts.

\begin{definition}\label{def: syntactic generalization}
	Let $C, D$ be \el concepts. Then
	$D$ is a \emph{syntactic generalization} of $C$ (written $C \sqsubsyn D$)
	if it is obtained from the reduced form of $C$ by replacing some occurrences of subconcepts $\neq\top$ with $\top$.
\end{definition}

For example, the concept $C = A_1\sqcap \exists r.(A_1\sqcap A_2)$ is already in reduced form, and its syntactic
generalizations include, among others, $\top\sqcap \exists r.(A_1\sqcap A_2) \equiv^\emptyset \exists r.(A_1\sqcap A_2)$, 
$A_1\sqcap \exists r.(\top \sqcap A_2) \equiv^\emptyset A_1\sqcap \exists r.A_2$, $\exists r.\top$, and $\top$.

\begin{lemma}\label{lem:sqsubsyn}
If $C\sqsubsyn D$, then $C\sub D$, and the length of any $\sqsubsyn$-chain issuing from $C$ is linearly bounded by the size of $C$.
\end{lemma}

\begin{proof}
We use a modified definition of size (called m-size) where only occurrences of concept and role names are counted.
Reducing a concept preserves equivalence and never increases the m-size. Since the concept constructors of $\el$ are monotonic,
$C\sqsubsyn D$ implies $C\subeq D$. In addition, the m-size of the reduced form of $C$ is strictly larger than the m-size
of the reduced form of $D$ since 
concepts $\neq\top$ have an m-size $> 0$ whereas $\top$ has m-size $0$.
This shows $C\not\equiv^\emptyset D$ (and thus $C\sub D$), since these reduced forms then cannot be equal up to associativity and commutativity of $\sqcap$.
In addition, it clearly yields the desired linear bound on the length of $\sqsubsyn$-chains.
\end{proof}

By Proposition~\ref{bounded:then:one:step:gen}, this linear bound implies that $\sqsubsyn$ is one-step generated. In the 
corresponding one-step relation $\sqsubsyn_1$, the replacements can be restricted to subconcepts that are concept names 
or existential restriction  of the form $\exists r.\top$. 
For example, we have (modulo equivalence)
$$\exists r.(A_1\sqcap A_2\sqcap A_3) \sqsubsyn_1 \exists r.(A_1\sqcap A_2)\sqsubsyn_1 \exists r.A_2 \sqsubsyn_1 \exists r.\top \sqsubsyn_1\top.$$
However, not all such restricted replacements lead to single steps w.r.t.\ $\sqsubsyn$. For example,
consider the concept $C = \exists r.(A_1\sqcap A_2) \sqcap \exists r.(A_2\sqcap A_3)$. Then replacing $A_3$ by $\top$ leads
to $D = \exists r.(A_1\sqcap A_2) \sqcap \exists r.(A_2\sqcap \top) \equiv \exists r.(A_1\sqcap A_2)$, but we have
$C\sqsubsyn \exists r.(A_1\sqcap A_2) \sqcap \exists r.A_3 \sqsubsyn D$.

Before proving that every $\sqsubsyn_1$-step can be realized by such restricted replacements, 
we use the fact that any $\el$ concept can be written as a conjunction of concept names and existential restrictions to give a recursive
characterization of $\sqsubsyn$. Let $C$ be an $\el$ concept, and assume that
its reduced form is 
$$
C' = A_1\sqcap\ldots\sqcap A_k\sqcap\exists r_1.C_1\sqcap \ldots\sqcap \exists r_\ell.C_\ell.
$$
Then we have $A_i\neq A_j$ for all $i \neq j$ in $\{1,\ldots,k\}$ and 
$r_\mu\neq r_\nu$ or $C_\mu\not\subeq C_\nu$ for all $\nu \neq \mu$ in $\{1,\ldots,\ell\}$, since otherwise
$C'$ would not be reduced. Replacing some occurrences of subconcepts with $\top$ then corresponds (modulo equivalence) to
\begin{itemize}
\item
   removing some of the conjuncts of the form $A_i$,
\item
   removing some of the conjuncts of the form $\exists r_\mu.C_\mu$,
\item
   replacing some of the conjuncts of the form $\exists r_\nu.C_\nu$ with a conjunct of the form $\exists r_\nu.D_\nu$
   where $C_\nu\sqsubsyn D_\nu$
\end{itemize}
such that at least one of these actions is really taken. Thus, $C\sqsubsyn_1 D$ implies that $D$ can be obtained from
the reduced form of $C$ by taking exactly one of these actions for exactly one conjunct. In fact, either taking several actions has the same
effect as taking one of them, or taking the actions one after another leads to a sequence of several strict syntactic generalizations steps,
which is precluded by the definition of $\sqsubsyn_1$.

\begin{lemma}
Let $C\not\equiv^\emptyset \top$ with reduced form $C' = A_1\sqcap\ldots\sqcap A_k\sqcap\exists r_1.C_1\sqcap \ldots\sqcap \exists r_\ell.C_\ell$,
and assume that $C\sqsubsyn_1 D$. Then $D$ is obtained (modulo equivalence) from $C'$ by either
\begin{enumerate}
\item\label{type:one} 
  removing exactly one of the concept names $A_i$,
\item\label{type:two}
  removing exactly one of the existential restrictions $\exists r_\mu.C_\mu$ for $C_\mu\equiv^\emptyset \top$, or
\item\label{type:three} 
  replacing exactly one of the existential restrictions $\exists r_\nu.C_\nu$ with $\exists r_\nu.D_\nu$\\ for $C_\nu\sqsubsyn_1 D_\nu$.
\end{enumerate}
\end{lemma}

\begin{proof}
As argued above, $C\sqsubsyn_1 D$ implies that $D$ is obtained from $C'$
by performing one of the following three actions:
\begin{itemize}
\item
   \emph{Removing exactly one of the conjuncts of the form $A_i$:} in this case, we are done.
\item
   \emph{Removing exactly one of the conjuncts of the form $\exists r_\mu.C_\mu$:} in this case we are done if
   $C_\mu\equiv^\emptyset \top$. Thus, assume that $C_\mu\not\equiv^\emptyset \top$. Let $D'$ be obtained from
   $C'$ by replacing $\exists r_\mu.C_\mu$ with $\exists r_\mu.\top$. Then we either have
   $C\sqsubsyn D'\sqsubsyn D$ or $D'\equiv^\emptyset D$. The first case contradicts our assumption
   that $C\sqsubsyn_1 D$. The second case is dealt with below since $C_\mu\sqsubsyn\top$.
\item
   \emph{Replacing exactly one of the conjuncts of the form $\exists r_\nu.C_\nu$ with a conjunct of the form $\exists r_\nu.D_\nu$
   where $C_\nu\sqsubsyn D_\nu$:} in this case we are done if $C_\nu\sqsubsyn_1 D_\nu$. Thus, assume that there is
   an $\el$ concept $D_\nu'$ such that $C_\nu\sqsubsyn D_\nu' \sqsubsyn D_\nu$. Since we already know that $\sqsubsyn$
   is one-step generated, we can assume without loss of generality that $C_\nu\sqsubsyn_1 D_\nu'$. Let $D'$ be obtained from
   $C'$ by replacing $\exists r_\nu.C_\nu$ with $\exists r_\mu.D_\nu'$. Then we either have
   $C\sqsubsyn D'\sqsubsyn D$ or $D'\equiv^\emptyset D$. The first case contradicts our assumption
   that $C\sqsubsyn_1 D$. In the second case, we are done.
\end{itemize}
Since there are no other cases, this completes the proof of the lemma.
\end{proof}

Based on this lemma, the following proposition can now easily be shown by induction on the role depth of $C$.

\begin{proposition}\label{syn:one:char}
Let $C$ be an $\el$ concept and $C'$ its reduced form. If $C \sqsubsyn_1 D$, then $D$ can be obtained (modulo equivalence) from
$C'$ by either replacing a concept name or a subconcept of the form $\exists r.\top$ by $\top$.
\end{proposition}

As an immediate consequence we obtain that $\sqsubsyn$ is effectively linearly branching.

\begin{corollary}\label{syn:poly:branch}
For a given $\el$ concept $C$, the set $\{D\mid C \sqsubsyn_1 D\}$ has a cardinality that is linear in the
size of $C$ and it can be computed in polynomial time.
\end{corollary}

\begin{proof}
That the cardinality of $\{D\mid C \sqsubsyn_1 D\}$ is linearly bounded by the size of $C$ is an immediate
consequence of Proposition~\ref{syn:one:char}. To compute the set, one first computes all concepts that
can be obtained by replacing in the reduced form of $C$ a concept name or a subconcept of the form $\exists r.\top$ by $\top$.
The polynomially many concepts obtained this way contain all the elements of $\{D\mid C \sqsubsyn_1 D\}$. Additional elements
in this set are obviously strictly subsumed by an element of $\{D\mid C \sqsubsyn_1 D\}$, and thus we can remove them
by removing elements that are not subsumption minimal.
\end{proof}

Now, we define our new weakening relation, which syntactically generalizes the right-hand sides of GCIs:
$$
\begin{array}{r@{\ \ \ }c@{\ \ \ }l}
{C\sqsubseteq D} \susy {C'\sqsubseteq D'} &\mbox{if}& C=C', D\sqsubsyn D'\  \mbox{and}\\
&&\{C'\sqsubseteq D'\}\not\models C\sqsubseteq D.
\end{array}
$$
The following theorem is an easy consequence of the properties of $\sqsubsyn$ and of 
Corollary~\ref{syn:poly:branch}.

\begin{theorem}\label{syn:weakening:rel:el}
        The relation $\susy$ on $\el$ axiom is a linear, complete, one-step generated, and effectively linearly branching weakening relation.
\end{theorem}

Due to fact that $\susy_1$-steps do not increase the size of axioms, the linear bounds on the branching of $\susy_1$ and the length
of $\susy$-chains imply that the algorithm described in the proof of
Proposition~\ref{max:strong:exist:prop} has an exponential search space.

\begin{corollary}
All maximally strong weakenings w.r.t.\ $\susy$ of an axiom in a justification can be computed in exponential time.
\end{corollary}

The following example shows that there may be exponentially many maximally strong weakenings w.r.t.\ $\susy$, and thus
the exponential complexity stated above is optimal.

\begin{example}\label{exp:max:strong:ex}
Let $\beta_i := P_i \sqcap Q_i \sqsubseteq B$ for $i = 1,\ldots, n$
and $\beta := A \sqsubseteq P_1\sqcap Q_1\sqcap \ldots \sqcap P_n\sqcap Q_n$.
We consider the ontology $\On = \Ons\cup\Onr$, where
$\Ons := \{\beta_i \mid 1\leq i\leq n\}$ and
$\Onr := \{\beta\}$.
Then $J = \{\beta\}$ is a justification for the consequence $\alpha = {A\sqsubseteq B}$,
and all axioms of the form $A\sqsubseteq X_1\sqcap X_2\sqcap \ldots \sqcap X_n$ with $X_i\in\{P_i,Q_i\}$
are maximally strong
weakenings w.r.t.\ $\susy$ of $\beta$ in $J$. The same is true for $\susu$ since in the absence of roles,
these two weakening relations coincide.
\end{example}

A single maximally strong weakening can however be computed in polynomial time.

\begin{proposition}
\label{thm:syn:poly}
A single maximally strong weakening w.r.t.\ $\susy$ can be computed in polynomial time.
\end{proposition}
\begin{proof}
The algorithm that computes a maximally strong weakening works as follows. Starting from the concept $D':=\top$, it
looks at all possible ways of making one step in the direction of $D$ using $\sqsupsyn_1$, i.e., it considers all
$D''$ where $D\sqsubsyneq D'' \sqsubsyn_1 D'$. The concepts $D''$ can be obtained by adding a concept name $A$
or an existential restriction $\exists r.\top$ at a place where (the reduced form of) $D$ has such a concept or restriction. Obviously, there are only polynomially
many such concepts $D''$. For each of them we check whether
$$
\Ons\cup (J \setminus \{C\sqsubseteq D\}) \cup \{C\sqsubseteq D''\} \models \alpha.
$$
If this is the case for all $D''$, we return $C\sqsubseteq D'$. Otherwise, we choose an arbitrary
$D''$ with
$
\Ons\cup (J \setminus \{C\sqsubseteq D\}) \cup \{C\sqsubseteq D''\} \not\models \alpha,
$
and continue with $D':= D''$.

This algorithm terminates after linearly many iterations since in each iteration the size of $D'$ is increased and it
cannot get larger than $D$. In addition, $C\sqsubseteq D'$ is maximally strong since for every axiom $C\sqsubseteq E$
such that ${C\sqsubseteq D} \susy {C\sqsubseteq E}\susy {C\sqsubseteq D'}$ there is a sequence
$E\sqsubsyn_1 \ldots\sqsubsyn_1 D''\sqsubsyn_1 D'$. Consequently, $C\sqsubseteq D''$ has the consequence,
and thus also $C\sqsubseteq E$.
\end{proof}

Nevertheless, we can show that deciding whether an axiom
is a maximally strong weakening w.r.t.\ $\susy$ is coNP-complete.

\begin{proposition}
The problem of
deciding whether a given \el GCI $C\sqsubseteq D'$ is a maximally strong weakening of the \el GCI $C\sqsubseteq D$ w.r.t.\
$\susy$ is coNP-complete.
\end{proposition}
\begin{proof}
First, we show the coNP upper bound. 
Let $\On = \Ons\cup\Onr$, $J\subseteq \Onr$ a justification of the consequence $\alpha$, $C\sqsubseteq D$ an element of $J$, and $C\sqsubseteq D'$ a GCI.
Obviously, we can decide in polynomial time whether ${C\sqsubseteq D}\susy {C\sqsubseteq D'}$ and whether 
$\Ons\cup (J\setminus\{{C\sqsubseteq D}\}) \cup \{{C\sqsubseteq D'}\}\not\models\alpha$. To disprove that $C\sqsubseteq D'$ is maximally strong,
we guess an $\el$ concept $D''$ such that $D\sqsubsyn D''\sqsubsyn D'$. This requires only polynomially many guesses: in fact, $D'$ is obtained from
$D$ by replacing linearly many occurrences of subconcepts with $\top$, and we simply guess which of these replacements are not done when going from
$D$ to $D''$. We then check in polynomial time whether $C\sqsubseteq D''$ satisfies 
\begin{itemize}
\item
  $\Ons\cup (J\setminus\{{C\sqsubseteq D}\}) \cup \{{C\sqsubseteq D''}\}\not\models\alpha$, and
\item
  $\{{C\sqsubseteq D'}\}\not\models {C\sqsubseteq D''}$.
\end{itemize}
If both tests succeed then $C\sqsubseteq D''$ is a counterexample to $C\sqsubseteq D'$ being maximally strong.

For the hardness proof, we use again the all-maximal-valuations problem. 
Given an instance $\varphi,\Vmc$ of the all-maximal-valuations problem, we construct an instance of our
problem as follows. For every subformula $\psi$ of $\varphi$, we introduce a new concept name $B_\psi$.
If $\psi$ is not a propositional variable, we define the TBox:
\[
\Tmc_\psi :=
        \begin{cases}
                \{ B_{\psi_1}\sqcap B_{\psi_2}\sqsubseteq B_\psi\} & \psi=\psi_1\land\psi_2 \\
                \{ B_{\psi_1}\sqsubseteq B_\psi, B_{\psi_2}\sqsubseteq B_\psi\} & \psi=\psi_1\lor\psi_2.
        \end{cases}
\]
Let $V$ be the set of all propositional variables appearing in $\varphi$, and let $csub(\varphi)$ be the set
of all subformulas of $\varphi$ that are not in $V$.
Define the concept 
$$
X_\Vmc := \bigsqcap_{W\in\Vmc} \exists r.\bigsqcap\{B_p\mid p\in W\}.
$$

We construct the ontology that has only one refutable axiom
$$X_\Vmc\sqsubseteq \exists r.\bigsqcap\{B_p\mid p\in V\},$$
and as static part the ontology
\[
\Tmc_s = \bigcup_{\psi\in csub(\varphi)}\Tmc_\psi \cup \{ \exists r. B_\varphi\sqsubseteq C\}
\]
Clearly, the refutable axiom is the only justification for $X_\Vmc\sqsubseteq C$.

For every valuation $W\subseteq V$, if $W$ is a subset of some valuation in \Vmc, then
\[
X_\Vmc \sqsubseteq \exists r.\bigsqcap\{B_p\mid p\in W\}\ \mbox{is equivalent to}\ X_\Vmc \sqsubseteq \top.
\]
We \emph{claim} that $X_\Vmc\sqsubseteq \top$ is a maximally strong weakening w.r.t.\ $\susy$ of the only refutable axiom
iff \Vmc is the set of all maximal valuations not satisfying $\varphi$.

To prove this claim, first assume that \Vmc is not the set of all maximal valuations not satisfying $\varphi$, i.e., there
is a maximal valuation $W$ not satisfying $\varphi$ such that $W\not\in \Vmc$. On the one hand, this implies that $W$ is incomparable w.r.t.\ inclusion
with any of the elements of $\Vmc$, and thus $X_\Vmc\sqsubseteq \exists r.\bigsqcap\{B_p\mid p\in W\}$ is not a tautology.
On the other hand, we have 
$$
\Tmc_s\cup\{X_\Vmc\sqsubseteq \exists r.\bigsqcap\{B_p\mid p\in W\}\}\not\models X_\Vmc\sqsubseteq C,
$$
and ${X_\Vmc\sqsubseteq \exists r.\bigsqcap\{B_p\mid p\in V\}}\susy{X_\Vmc\sqsubseteq \exists r.\bigsqcap\{B_p\mid p\in W\}}$.
This shows that the tautology $X_\Vmc\sqsubseteq \top$ is not a maximally strong weakening w.r.t.\ $\susy$ of the only refutable axiom
$X_\Vmc\sqsubseteq \exists r.\bigsqcap\{B_p\mid p\in V\}$.

Conversely, assume that \Vmc is the set of all maximal valuations not satisfying $\varphi$, and that $\gamma$ is a
maximally strong weakening w.r.t.\ $\susy$ of $X_\Vmc\sqsubseteq \exists r.\bigsqcap\{B_p\mid p\in V\}$. If $\gamma = {X_\Vmc\sqsubseteq \top}$,
then we are done. Otherwise, there is a set $W\subseteq \Vmc$ such that $\gamma = {X_\Vmc\sqsubseteq \exists r.\bigsqcap\{B_p\mid p\in W\}}$.
But then $\Tmc_s\cup\{\gamma\}\not\models X_\Vmc\sqsubseteq C$ implies that $W$ does not satisfy $\varphi$, and thus $W$ is a subset of some valuation in \Vmc.
Consequently, $\gamma$ is a tautology and thus equivalent to $X_\Vmc \sqsubseteq \top$. This shows that $X_\Vmc \sqsubseteq \top$
is a maximally strong weakening w.r.t.\ $\susy$ of $X_\Vmc\sqsubseteq \exists r.\bigsqcap\{B_p\mid p\in V\}$.
\end{proof}


\section{Conclusions}
\label{sec:conclusion}

We have introduced a framework for repairing DL-based ontologies that is based
on weakening axioms rather than deleting them, and have shown how to instantiate
this framework for the DL \el using appropriate weakening relations. More precisely,
we have introduced weakening relations of decreasing strength
${\succ_g}\supset {\succ_s}\supset {\susu} \supset {\susy}$, and have shown
that ${\succ_g}$ and ${\succ_s}$ do not satisfy the properties required to apply
our gentle weakening approach. In contrast, both ${\susu}$ and ${\susy}$ satisfy these
properties, but from a complexity point of view ${\susy}$ is to be preferred.
 
Computing maximally strong weakenings w.r.t.\ ${\susu}$ or ${\susy}$ using the algorithm described in the proof of
Proposition~\ref{max:strong:exist:prop} is akin to the black-box approach for computing justifications.
It would be interesting to see whether a glass-box approach that modifies an \el reasoning procedure
can also be used for this purpose, similar to the way a tableau-based algorithms for \alc
was modified in \cite{Lam08}. This should be possible for ${\susy}$, whereas handling ${\susu}$ with
a glass-box approach is probably more challenging, but might yield better complexity upper bounds
than the generic approach based on Proposition~\ref{max:strong:exist:prop}.
 
Our weakening relations can also be used in the setting where the ontology is first modified, and then
repaired using the classical approach as in \cite{DuQF14}. In fact, for effectively finitely branching and well-founded
weakening relations such as $\susu$ and $\susy$, we can add for each axiom all (or some of) its finitely many weakenings 
w.r.t.\ the given relation, and then apply the classical repair approach. In contrast to the gentle repair
approach proposed in this paper, a single axiom could then be replaced by several axioms, which might blow up
the size of the ontology.

In order to apply our gently repair approach in practice, one can either compute all maximally strong weakening,
and let the user choose between them, which should be viable at least for $\susy$. Alternatively, one can try
to find heuristics for obtaining weakening oracles that compute ``good'' weakenings or involve the user in
the decisions made in each weakening step.

\subsection*{Acknowledgments}
This work is partially supported by DFG within the Research Training Group 1907 (RoSI).
The authors would like to thank Bernhard Ganter for helpful discussion regarding
one-step generated pre-orders.



\end{document}